\documentclass[11pt]{article}

\usepackage[utf8x]{inputenc}
\usepackage[T1]{fontenc}
\usepackage[letterpaper,textheight=58em]{geometry}
\PassOptionsToPackage{hyphens}{url}\usepackage[colorlinks,citecolor=blue,
urlcolor=blue]{hyperref}

\usepackage{amsmath,amsfonts, amssymb, dsfont, amsthm}
\usepackage{mathabx}
\usepackage{enumerate}
\usepackage{graphicx}
\usepackage{tikz}
\usetikzlibrary{calc,shapes,backgrounds,arrows,automata,shadows,positioning}
\usepackage{xcolor}
\usepackage{float}
\usepackage{multirow}
\usepackage{array}
\usepackage{url}
\usepackage{xspace}
\usepackage{hhline}
\usepackage{ctable}

\newif\ifdoublespacing
\doublespacingtrue
\ifdoublespacing
  \usepackage{setspace}
\fi

\usepackage[textwidth=8em,textsize=small]{todonotes}

\usepackage[round]{natbib}
\setcitestyle{authoryear,open={(},close={)},citesep={,}}
\usepackage{lineno}

\newtheorem{proposition}{Proposition}
\newtheorem{corollary}{Corollary}

\renewcommand{\d}{\text{d}}
\newcommand{\Beta}{{B}}
\newcommand{\Cov}{{\mathbb C}\text{ov}}

\newcommand{\Esp}{\xspace\mathbb E}

\newcommand{\Ncal}{\mathcal{N}}
\newcommand{\off}{\textrm{off}}

\newcommand{\Pcal}{\mathcal{P}}
\newcommand{\Qcal}{\mathcal{Q}}

\newcommand{\Sbar}{S_+}
\renewcommand{\sp}{\text{struct}}
\newcommand{\stab}{\text{stab}}

\newcommand{\Var}{\mathbb V}

\newcommand{\mb}[1]{\mathbf{#1}}
\newcommand{\bs}[1]{\boldsymbol #1}

\newcommand{\matr}[1]{{#1}}

\newcommand{\trans}{\intercal}
\newcommand{\transpose}[1]{\matr{#1}^\trans}

\DeclareMathOperator*{\minimize}{minimize}

\DeclareMathOperator{\tr}{tr}

\DeclareMathOperator{\diag}{diag} 
\DeclareMathOperator{\mtov}{vec} 
\newcommand{\kro}{\otimes} 
\newcommand{\had}{\odot}   

\newcommand{\nodesize}{2em}
\newcommand{\edgeunit}{2.5*\nodesize}
\tikzstyle{hidden}=[draw, circle, fill=gray!50, minimum width=\nodesize, inner sep=0]
\tikzstyle{observed}=[draw, circle, minimum width=\nodesize, inner sep=0]
\tikzstyle{eliminated}=[draw, circle, minimum width=\nodesize, color=gray!50, inner sep=0]
\tikzstyle{empty}=[]
\tikzstyle{arrow}=[->, >=latex, line width=1pt]
\tikzstyle{edge}=[-, line width=1pt]
\tikzstyle{dashedarrow}=[->, >=latex, dashed, line width=1pt]
\tikzstyle{lightarrow}=[->, >=latex, line width=1pt, fill=gray!50, color=gray!50]

\newcommand{\PLN}{PLN\xspace}
\newcommand{\PLNnetwork}{PLN-network\xspace}

\title{Variational inference for sparse network reconstruction from count data}
\author{Julien Chiquet, Mahendra Mariadassou, Stéphane Robin} 
\date{UMR MIA-Paris, AgroParisTech,  INRA, Universit\'e Paris-Saclay, 75005 
Paris, France \\
MaIAGE, INRA, Université Paris-Saclay, 78350, Jouy-en-Josas, France}

\begin{document}

\maketitle

\begin{abstract}%
  \noindent In multivariate statistics, the question of finding direct
  interactions can be formulated as a problem of network inference -
  or network reconstruction - for which the Gaussian graphical model
  (GGM) provides a canonical framework. Unfortunately, the Gaussian
  assumption does not apply to count data which are encountered in
  domains such as genomics, social sciences or ecology.

  To circumvent this limitation, state-of-the-art approaches use
  two-step strategies that first transform counts to pseudo Gaussian
  observations and then apply a (partial) correlation-based approach
  from the abundant literature of GGM inference.  We adopt a different
  stance by relying on a latent model where we directly model counts
  by means of Poisson distributions that are conditional to latent
  (hidden) Gaussian correlated variables.  In this multivariate
  Poisson lognormal-model, the dependency structure is completely
  captured by the latent layer.  This parametric model enables to
  account for the effects of covariates on the counts.
 
  To perform network inference, we add a sparsity inducing constraint
  on the inverse covariance matrix of the latent Gaussian
  vector. Unlike the usual Gaussian setting, the penalized likelihood
  is generally not tractable, and we resort instead to a variational
  approach for approximate likelihood maximization. The corresponding
  optimization problem is solved by alternating a gradient ascent on
  the variational parameters and a graphical-Lasso step on the
  covariance matrix.

  We show that our approach is highly competitive with the existing
  methods on simulation inspired from microbiological data. We then
  illustrate on three various data sets how accounting for sampling
  efforts via offsets and integrating external covariates (which is
  mostly never done in the existing literature) drastically changes
  the topology of the inferred network.
  \\

\end{abstract}

\section{Introduction} \label{sec:intro}

Networks are the \emph{de facto} mathematical object used to model and
represent pairwise interactions between entities of interest. Examples
include air traffic between airports, social interactions between
participants of a conference, trophic relationships between species,
gene regulations, ecological interactions between microbial species,
etc. However, most networks are not observed directly but must be
reconstructed first from indirect node-level observations using some
kind of statistical procedure. In this perspective, graphical models
are popular among statisticians to explore relationships between nodes
in graphs since undirected graphical models \citep{Lau96}, also called
Markov random fields \citep{Harris2016}, are a convenient class of
models with sound theoretical groundings for capturing conditional
dependence relationships between nodes: $i$ and $j$ are linked in
$\mb{G}$ ($i \sim j$) if and only if features $i$ and $j$ are
conditionally dependent given all the others.  Powerful inference
procedures exist for Gaussian Graphical Models (GGM) for continuous
data and Ising or voter models for binary data and is still a very
active field of research. An informative and non-exhaustive set of
seminal papers in this field may include
\citet{yuan2007,banerjee2008,ravikumar2010high,Meinshausen2006,cai2011constrained,khare2015convex}. On
the application side, GGM have been successfully used in many fields,
most notably biology, to understand complex genetic regulations
\citep{Moignard2015, Fiers2018}, to identify direct contacts between
protein subunits \citep{Drew2017} or to identify functional pathways
associated to a disease \citep{Yu2015}. Unfortunately, we lack such
powerful estimation procedures for non-Gaussian data, especially count
data, which is the focus of this work.

Count data arise naturally in fields such as ecology (species count at
a given site), transcriptomics (copy number of a transcript in a
tissue) and quite broadly, all subfields of biology based on molecular
markers and high-througput sequencing.  They also arise in political
sciences (voting outcomes), tourism management (number of visitors to
sightseeing spots), to cite only a few.  By analogy to the Gaussian
graphical setting, many efforts have been devoted throughout the years
to develop multivariate Poisson distribution in order to model
dependencies \citep[see][for a review]{IYA16}, since Poisson is the
natural probability distribution for modeling counts.  Unfortunately,
there is no satisfying Poisson counterpart to the multivariate
Gaussian. \citet{Besag1974} introduced Poisson Graphical Model (PGM)
and proved that PGM can only capture negative dependencies to ensure
consistency of the joint distribution. \citet{yang2012graphical}
proposed variants of \citet{Besag1974}'s PGM but none of them was
completely satisfying. Usually, they failed to have either marginal or
conditional Poisson distributions. \citet{allen2012log} also proposed
a local PGM satisfying the local Markov property but do not have a
joint consistent graphical model. In the same vein,
\citet{Gallopin2013} considered log-normal models. In both methods,
authors estimate the neighborhood of a node by performing a
generalized linear regression \emph{à la}
\citet{Meinshausen2006}. Another common yet more recent approach --
used for microbial ecology in SPIEC-EASI \citep{Kurtz2015} and BAnoCC
\citep{Schwager2017} -- addresses the problem differently, by $i)$
replacing counts with (regularized) frequencies, and $ii)$ taking
their log-ratios before $iii)$ moving back to the GGM framework. A
positive side effect of this transformation is to remedy the issue
referred to as the \emph{compositionality problem}: counts can only be
compared to each other within a sample but not across samples as they
depend on a sample-specific size-factor, which may induce spurious
negative correlations of its own. This problem is particularly acute
in molecular biology where counts are constrained by the sampling
effort (\emph{e.g.} sequencing depth). Note, however, that the
count-to-frequency transformations prevents one from integrating
heterogeneous sources of count data (\emph{e.g.} bacteria and fungi in
ecology, gene expression and methylation levels in functional
genomics) and to find interactions between nodes of different natures,
although they are known to be important in certain contexts
\citep{Lima-Mendez2015}. Finally, a common shortcoming of the two
families of approaches (PGM and preprocessed GGM), at least in their
vanilla formulation, is that they do not offer a systematic way to
control for covariates and confounding factors: differences in mean
counts induced by differences in a structuring factor (\emph{e.g.}
nutrient availibility in ecology) may be mistakenly inferred as
interactions \citep{Vacher2016}.

In this paper, we tackle the limitations mentioned above by recourse
to a hierachical Poisson log-normal (PLN) model with a latent Gaussian
layer and an observed Poisson layer. We use the GGM formulation to
model direct interactions between features in the Gaussian layer and
include covariates in the Poisson layer to control for confounding
factors, as was done by several authors in different contexts
\citep[see][]{ChG95,PaL07,MKD08}.  Finally, we address the
compositionality problem by using offsets \citep{Agr96} and can thus
reconstruct interaction networks on heterogeneous groups of features
observed in the same samples but using different techniques. The model
is similar to the one introduced in \citet{Biswas2016} but the
inference is significantly different and has a deeper statistical
grounding. In particular, and unlike \citet{Biswas2016}, we consider
the latent variable as a random variable and not as a parameter. We
therefore use a variational inference procedure to estimate the
interaction network. The resulting optimization procedure is more
complex but accounts for the uncertainty of the latent variables.

The manuscript is organized as follows: Section~\ref{sec:model}
introduces notation and the PLN model. Section~\ref{sec:inference}
presents the variational approximation and inference
procedure. Section~\ref{sec:simul} presents the results of a
simulation study and Section~\ref{sec:appli} shows networks
reconstructed from three real world count datasets: two originating
from community ecology and one from voting outcomes in a recent French
election.

\section{A Graphical Model for Multivariate Count Data} \label{sec:model}


\subsection{Multivariate Poisson Log-Normal (\PLN) Model}

We first remind the definition of the multivariate \PLN model
\citep{AiH89}. The model involves parameters
$\mb{\mu} = (\mu_j)_{1 \leq j, \leq p}$ and
$\mb{\Sigma} = (\sigma_{jk})_{1 \leq j, k \leq p}$. An \textit{i.i.d.}
\PLN sample is drawn as follows: for each observed $p$-dimensional
count vector $Y_i$ ($1 \leq i \leq n$), a Gaussian latent
(\textit{i.e.} hidden) $p$-dimensional vector $Z_i$ is drawn and the
coordinates of $Y_i$ are sampled independently from a Poisson
distribution, conditionally on $Z_i$:
\begin{equation} \label{eq:PLNmodel}
  \begin{array}{r@{\hspace{4ex}}l}
    (Z_i)_{1 \leq i \leq n}\text{ iid},  
    &  
    Z_i \sim \Ncal(\bs{0}_p, \mb{\Sigma}), \\[2ex]
    (Y_{ij})_{1 \leq i \leq n, 1 \leq j \leq p} \text{ indep.} | \; Z_{ij}, 
    & 
    Y_{ij} \; | \; Z_{ij} \sim \Pcal(\exp\{\mu_j + Z_{ij}\}).
\end{array}
\end{equation}
In the following, all count vectors $Y_i$ are gathered into the $n \times p$ 
matrix $\mb{Y} \triangleq (Y_{ij})_{1 \leq i \leq n, 1 \leq j \leq p}$. The $n 
\times p$ matrix $\mb{Z}$ is defined as $\mb{Z} \triangleq 
(Z_{ij})_{1 \leq i \leq n, 1 \leq j \leq p}$ in the same way.
The \PLN distribution displays several interesting properties such as over-dispersion with respect to the Poisson distribution:
\begin{equation}\label{eq:PLN_properties}
\Esp(Y_{ij}) = e^{\mu_j + \sigma_{jj}/2}, \qquad \Var(Y_{ij}) = \Esp(Y_{ij}) + (e^{\sigma_{jj}} - 1) \Esp(Y_{ij})^2 \geq \Esp(Y_{ij})
\end{equation}
and arbitrary sign for the covariance between the coordinates:
\begin{displaymath}
  \text{for } j\neq k, \quad \Cov(Y_{ij}, Y_{ik}) = (e^{\sigma_{jk}} - 1) \Esp(Y_{ij}) \Esp(Y_{ik}),
\end{displaymath}
that is: $\Cov(Y_{ij}, Y_{ik})$ has the same sign as $\Cov(Z_{ij}, Z_{ik}) = \sigma_{jk}$. 

\paragraph{Introducing covariates.}
Interestingly, covariates can be easily introduced in the \PLN model,
replacing the constant vector $\mu$ with a regression
term. Furthermore, in many applications dealing with counts, it is
desirable to introduce an offset term to account for some known effect
such as the sampling effort. Denote
$x_i = (x_{i\ell})_{1 \leq \ell \leq d}$ the vector of covariates for
observation $i$ and
$\mb{\Beta} = (\beta_{\ell j})_{1 \leq \ell \leq d, 1 \leq j \leq p}$
the corresponding matrix of regression coefficients. Also denote by 
$o_{ij}$ the offset term for count $Y_{ij}$. Both can be accounted for
by modifying the distribution of the count $Y_{ij}$ given in
\eqref{eq:PLNmodel} into
\begin{equation}
  \label{eq:PLN_covariates}
  Y_{ij} \; | \; Z_{ij} \sim \Pcal\left(\exp\{o_{ij} + x_i^\trans \beta_j + Z_{ij}\}\right).
\end{equation}
We further define the offset matrix
$\mb{O} = (o_{ij})_{1 \leq i \leq n, 1 \leq j \leq p}$ and the design
matrix $\mb{X} = (x_{i\ell})_{1 \leq i \leq n, 1 \leq \ell \leq d}$.

\begin{center}
  \S
\end{center}

The \PLN model is actually quite general and can be used for many
purposes. \cite{CMR17} show how probabilistic PCA can be casted in
this framework to perform dimension reduction. The following section
shows how graphical models fit within the PLN model.

\subsection{The \PLNnetwork graphical model}

In this work, we are interested in modeling the dependency structure that 
relates the coordinates of the count vectors $Y_i$. As mentioned in Section 
\ref{sec:intro}, no generic multivariate model is available for counts and 
existing models often impose undesired constraints on the dependency structure. 
To circumvent this issue, we use the \PLN model to push the structure inference 
problem to the latent space and to infer the dependency structure relating the 
coordinates of the latent vector $Z_i$. 

We use the framework of graphical models \citep{Lau96} to model this
dependency structure. Intuitively, the graph encodes the conditional
dependence structure between random variables. Formally, $Z_i$ and
$Z_j$ are connected in the graph if and only $Z_i$ and $Z_j$ are
independent conditionally on all other variables, that is:
$Z_i \not\perp Z_j \mid Z_{\setminus\{i,j\}}$. Now, because the
$Z_i$'s are jointly Gaussian, so is $(Z_i, Z_j | Z_{-\{i,j\}})$. In
particular, the partial correlation between $Z_i$ and $Z_j$ given the
$(Z_k)_{k \neq i,j}$ is
$\rho_{ij} = - \Omega_{ij} / \sqrt{\Omega_{ii}\Omega_{jj}}$ where
$\mb{\Omega} \triangleq \mb{\Sigma}^{-1}$ is the precision
matrix. Therefore $Z_i$ and $Z_j$ are conditionally independent if and
only if $\Omega_{ij} = 0$ and the structure inference problem reduces
to the determination of the support of $\mb{\Omega}$. This precision
matrix is assumed to be sparse.  In this perspective, it is critical
to account for covariates
that may have an effect on the observed counts to avoid spurious edges
in the inferred graphical model \citep[see e.g.][and
discussions]{CPW12}. As a consequence, in this paper, we adopt the
following parametrization of the \PLN model:
\begin{equation} \label{eq:PLNnetwork}
  \begin{array}{r@{\hspace{4ex}}l}
    (Z_i)_{1 \leq i \leq n}\text{ iid},  & Z_i \sim \Ncal(\bs{0}_p, \mb{\Omega}^{-1}), 
    \qquad \mb{\Omega} \text{ sparse}, \\[2ex]
  (Y_{ij})_{1 \leq i \leq n, 1 \leq j \leq p} \; \text{ indep.} \; | \; Z_{ij}, 
  &
  Y_{ij} \; | \; Z_{ij} \sim \Pcal(\exp\{o_{ij} + x_i^\trans \beta_j + Z_{ij}\}),
\end{array}
\end{equation}
which separates the structure parameter $\mb{\Omega}$ from the other
effect parameters $\mb{O}$ and $\mb{\Beta}$. We emphasize that pushing
the structure inference problem from the observed space of the $Y_i$
to the latent space of the $Z_i$ has some consequences. Indeed, it can
be easily checked that, if the graphical model of the $Z_i$ is
connected (that is, if no subset of latent coordinates is separated
from the rest), then all count coordinates are correlated, so that the
graphical model of the marginal distribution of the $Y_i$ is fully
connected (see Figure \ref{fig:graphicalmodel}, top). Only a
separation in the latent space will result in a separation in the
observed space (see Figure \ref{fig:graphicalmodel}, bottom). The
inference framework we propose must be therefore interpreted as
follows: all the dependency is captured in the latent space and the
lower Poisson layer in \eqref{eq:PLNnetwork} models an independent
measurement noise.
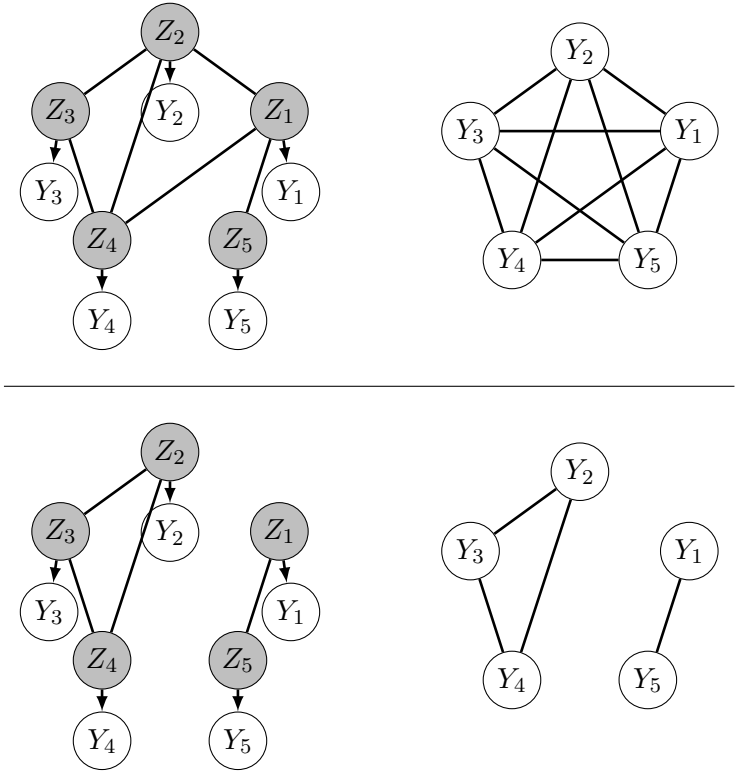
\begin{figure}[ht]
  \begin{center}
\begin{tabular}{ccc}
  \begin{tikzpicture}[scale=.8]
    \node[hidden] (Z1) at ( 0.95*\edgeunit,  0.31*\edgeunit) {$Z_1$};
    \node[hidden] (Z2) at (-0.00*\edgeunit,  1.00*\edgeunit) {$Z_2$};
    \node[hidden] (Z3) at (-0.95*\edgeunit,  0.31*\edgeunit) {$Z_3$};
    \node[hidden] (Z4) at (-0.59*\edgeunit, -0.81*\edgeunit) {$Z_4$};
    \node[hidden] (Z5) at ( 0.59*\edgeunit, -0.81*\edgeunit) {$Z_5$};
    
    \draw[edge] (Z1) to (Z2);  \draw[edge] (Z1) to (Z4);  
    \draw[edge] (Z1) to (Z5);  \draw[edge] (Z2) to (Z3);  
    \draw[edge] (Z2) to (Z4);  \draw[edge] (Z3) to (Z4); 

    \node[observed] (Y1) at ( 1.05*\edgeunit, -0.39*\edgeunit) {$Y_1$};
    \node[observed] (Y2) at (-0.00*\edgeunit,  0.30*\edgeunit) {$Y_2$};
    \node[observed] (Y3) at (-1.05*\edgeunit, -0.39*\edgeunit) {$Y_3$};
    \node[observed] (Y4) at (-0.59*\edgeunit, -1.51*\edgeunit) {$Y_4$};
    \node[observed] (Y5) at ( 0.59*\edgeunit, -1.51*\edgeunit) {$Y_5$};

    \draw[arrow] (Z1) to (Y1); 
    \draw[arrow] (Z2) to (Y2);
    \draw[arrow] (Z3) to (Y3);
    \draw[arrow] (Z4) to (Y4);
    \draw[arrow] (Z5) to (Y5);
  \end{tikzpicture}
  & \qquad \qquad &
  \begin{tikzpicture}[scale=.8]
  \node[observed] (Y1) at ( 0.95*\edgeunit,  0.31*\edgeunit) {$Y_1$};
  \node[observed] (Y2) at (-0.00*\edgeunit,  1.00*\edgeunit) {$Y_2$};
  \node[observed] (Y3) at (-0.95*\edgeunit,  0.31*\edgeunit) {$Y_3$};
  \node[observed] (Y4) at (-0.59*\edgeunit, -0.81*\edgeunit) {$Y_4$};
  \node[observed] (Y5) at ( 0.59*\edgeunit, -0.81*\edgeunit) {$Y_5$};
  \node[empty] (YY) at ( 0.59*\edgeunit, -1.51*\edgeunit) {};

  \draw[edge] (Y1) to (Y2);  \draw[edge] (Y1) to (Y3);  
  \draw[edge] (Y1) to (Y4);  \draw[edge] (Y1) to (Y5);  
  \draw[edge] (Y2) to (Y3);  \draw[edge] (Y2) to (Y4); 
  \draw[edge] (Y2) to (Y5);  \draw[edge] (Y3) to (Y4);  
  \draw[edge] (Y3) to (Y5);  \draw[edge] (Y4) to (Y5);  
  \end{tikzpicture}
  \\[2ex] \hline \\
  \begin{tikzpicture}[scale=.8]
  \node[hidden] (Z1) at ( 0.95*\edgeunit,  0.31*\edgeunit) {$Z_1$};
  \node[hidden] (Z2) at (-0.00*\edgeunit,  1.00*\edgeunit) {$Z_2$};
  \node[hidden] (Z3) at (-0.95*\edgeunit,  0.31*\edgeunit) {$Z_3$};
  \node[hidden] (Z4) at (-0.59*\edgeunit, -0.81*\edgeunit) {$Z_4$};
  \node[hidden] (Z5) at ( 0.59*\edgeunit, -0.81*\edgeunit) {$Z_5$};
  
  \draw[edge] (Z1) to (Z5);  \draw[edge] (Z2) to (Z3);  
  \draw[edge] (Z2) to (Z4);  \draw[edge] (Z3) to (Z4); 

  \node[observed] (Y1) at ( 1.05*\edgeunit, -0.39*\edgeunit) {$Y_1$};
  \node[observed] (Y2) at (-0.00*\edgeunit,  0.30*\edgeunit) {$Y_2$};
  \node[observed] (Y3) at (-1.05*\edgeunit, -0.39*\edgeunit) {$Y_3$};
  \node[observed] (Y4) at (-0.59*\edgeunit, -1.51*\edgeunit) {$Y_4$};
  \node[observed] (Y5) at ( 0.59*\edgeunit, -1.51*\edgeunit) {$Y_5$};
  
  \draw[arrow] (Z1) to (Y1); 
  \draw[arrow] (Z2) to (Y2);
  \draw[arrow] (Z3) to (Y3);
  \draw[arrow] (Z4) to (Y4);
  \draw[arrow] (Z5) to (Y5);
  \end{tikzpicture}
  & \qquad \qquad &
  \begin{tikzpicture}[scale=.8]
  \node[observed] (Y1) at ( 0.95*\edgeunit,  0.31*\edgeunit) {$Y_1$};
  \node[observed] (Y2) at (-0.00*\edgeunit,  1.00*\edgeunit) {$Y_2$};
  \node[observed] (Y3) at (-0.95*\edgeunit,  0.31*\edgeunit) {$Y_3$};
  \node[observed] (Y4) at (-0.59*\edgeunit, -0.81*\edgeunit) {$Y_4$};
  \node[observed] (Y5) at ( 0.59*\edgeunit, -0.81*\edgeunit) {$Y_5$};
  \node[empty] (YY) at ( 0.59*\edgeunit, -1.51*\edgeunit) {};

  \draw[edge] (Y1) to (Y5);  \draw[edge] (Y2) to (Y3);  
  \draw[edge] (Y2) to (Y4);  \draw[edge] (Y3) to (Y4);  
  
  \end{tikzpicture}
  \end{tabular}
  \caption{Two examples (top/bottom) of the \PLNnetwork graphical
    representation. Left: joint distribution of $p(Z_i, Y_i)$. Right:
    marginal distribution $p(Y_i)$.  The graph on the top right is a
    clique because the graph of the $Z_i$'s on the top left is
    connected.}
  \label{fig:graphicalmodel}
 \end{center}
\end{figure}


\section{Sparse Variational Inference} \label{sec:inference}

We now describe the inference strategy adopted for Model
\eqref{eq:PLNnetwork}. The aim is primarily to provide an estimate of
the parameter $\bs{\theta} = (\mb{\Beta}, \mb{\Omega})$.

\subsection{Incomplete data model}

Model \eqref{eq:PLNnetwork} belongs to the class of incomplete data
model, as the latent vectors $Z_i$ are unobserved. Therefore the
evaluation of the log-likelihood of the observed data
$ \log p_{\bs{\theta}}(\mb{Y}) = \log \int p_{\bs{\theta}}(\mb{Y},
\mb{Z}) \d \mb{Z} $
is often intractable, as well as its maximization with respect to
$\bs{\theta}$.  In this setting, the most popular strategy to perform
maximum likelihood is to use the EM algorithm of \cite{DLR77}, which
requires the evaluation of the conditional expectation of the complete
log-likelihood
$\Esp_{\bs{\theta}} \left[ \log p_{\bs{\theta}}(\mb{Y}, \mb{Z}) |
  \mb{Y} \right]$.
Unfortunately, this amounts to compute (some moments of) the
conditional distribution of each latent vector $Z_i$ conditionally to
the corresponding count vector $Y_i = (Y_{ij})_{1 \leq j \leq p}$,
which has no close form in the PLN model. \cite{Kar05} suggests to
achieve this task via numerical or Monte-Carlo integration, but this
approach is computationally too demanding when dealing even with a
moderate number of variables.

\paragraph{Variational approximation.} To circumvent this issue, we resort to a variational approximation \citep{WaJ08}, which consists in finding a proxy for the conditional distribution $p_{\bs{\theta}}(Z_i|Y_i)$. This approach relies on a divergence measure between the true conditional distribution and the approximated distribution, chosen within a simple class of distributions $\Qcal$. 

In this paper we choose $\Qcal$ as the set of Gaussian distributions. Namely, each conditional distribution $p_{\bs{\theta}}(Z_i|Y_i)$ is approximated with a multivariate Gaussian distribution with mean vector $\mb{m}_i$ and diagonal covariance matrix $\mb{S}_i = \diag(\mb{s}_i^2)$. As a consequence, the approximate distribution $q$ is fully parametrized by $\bs{\psi} = (\mb{M}, \mb{S})$, where $\mb{M}= [\mb{m}_1^\trans \dots \mb{m}_n^\trans]^\trans$,  $\mb{S} = [(\mb{s}^2_1)^\trans \dots (\mb{s}^2_n)^\trans]^\trans$ and $\Qcal$ is defined by
\begin{equation} \label{eq:var_distr}
\Qcal = \left\{q: q_{\bs{\psi}}(\mb{Z}) = \prod_{i=1}^n \Ncal(Z_i; \mb{m}_i, \mb{S}_i) = \prod_{i=1}^n q_i(Z_i) \right\}.
\end{equation}
We emphasize that the vectors $Z_i$ are independent conditionally on the $Y_i$'s, so the approximation does not lie in the product form but only in the Gaussian form of each approximate distribution. 

Choosing the Kullback-Leibler divergence to measure the quality of the approximation leads to the ``variational'' EM (VEM) algorithm, which aims to maximize the lower bound of the log-likelihood of the observed data. This lower bound is defined by
\begin{equation} \label{eq:var_bound} 
  \begin{aligned}
   J(\mb{Y} ; \bs{\psi}, \theta) & \triangleq \log p_{\bs{\theta}}(\mb{Y}) - KL\left[q_{\bs{\psi}}(\mb{Z})||p_{\bs{\theta}}(\mb{Z}|\mb{Y})\right]\\
   & = \Esp_q \left[ \log p_{\bs{\theta}}(\mb{Y}, \mb{Z}) \right] - \Esp_q \left[ \log q_{\bs{\psi}}(\mb{Z}) \right],
  \end{aligned}
\end{equation}
where $\Esp_q$ stands for the expectation with respect to the distribution $q_{\bs{\psi}}$.  

\paragraph{Sparse structure inference.} To infer the structure -- that
is the underlying 'network' -- we need to determine the support of
$\mb{\Omega}$. To this end we add an $\ell_1$ sparsity inducing
penalty to the lower bound of the likelihood, mimicking the Gaussian
case like in the Graphical-Lasso. The corresponding objective
function that we suggest to maximize is thus
\begin{equation} \label{eq:Jsp}
 J_\sp(\mb{Y}; \bs{\psi}, \bs{\theta}) \triangleq  J(\mb{Y}; \bs{\psi}, \bs{\theta}) - \lambda \; \| \mb{\Omega} \|_{\ell_1, \off} \leq \log p_{\bs{\theta}}(\mb{Y}) - \lambda \; \|\mb{\Omega}\|_{\ell_1, \off},
\end{equation}
where $\|\mb{\Omega}\|_{\ell_1, \off} = \sum_{j \neq k} |\Omega_{jk}|$
is the off-diagonal $\ell_1$-norm of $\mb{\Omega}$ and $\lambda>0$ is
a tuning parameter controling the amount of sparsity. Note that, by
construction, $J_\sp$ is a lower bound of the penalized
log-likelihood.

\subsection{Inference algorithm}

\paragraph{Objective function.}  The properties of the objective
function $J_\sp$ are mainly inherited from the properties of $J$ since
they only differ by the sparsity-inducing penalizing term, so we first
discuss $J$: by Definition \eqref{eq:var_bound} of the unpenalized
variational lower bound and thanks to the form \eqref{eq:var_distr} of
the approximate distribution, we have
\begin{equation*}
  J(\mb{Y} ; \bs{\psi}, \bs{\theta})
  = \sum_{i=1}^n \Esp_{q_i} \left[\log p_{\bs{\theta}}(Y_i|Z_i)\right] + \Esp_{q_i} \left[\log p_{\bs{\theta}}(Z_i)\right] - \Esp_{q_i} \left[\log q_{\bs{\psi}} (Z_i) \right].
\end{equation*}
Derivation of a close form is then straightforward by means of basic
properties of the multivariate Gaussian distribution and of the PLN
distribution (see \eqref{eq:PLN_properties}). We first need a couple
of auxiliary matrices, namely $\mb{\Sbar} = \sum_{i=1}^n \mb{S}_i$,
the accumulated variance matrix;
$\hat{\mb{\Sigma}} = n^{-1} \left(\mb{M}^\trans\mb{M} + \mb{\Sbar}
\right)$,
the estimated covariance matrix and
$\mb{A} \triangleq (A_{ij})_{1\leq i\leq n, 1\leq j \leq p}$ the
$n\times p$ matrix of expected counts, the entries of which are
defined by
\begin{equation*}
  A_{ij} \triangleq \Esp_q \left(Y_{ij} \right) = \Esp_q \left(\exp(o_{ij} + 
x_i^\trans \beta_j + Z_{ij})\right) = \exp(o_{ij} + x_i^\trans
  \beta_j + m_{ij} + s^2_{ij}/2).
\end{equation*}
These quantities allows us to write a compact form of the approximated
log-likelihood:
\begin{multline}
  \label{eq:J_matrix_form}
  J(\mb{Y} ; \bs{\psi}, \bs{\theta})
  =  \transpose{\mb{1}}_{n}\bigg(\mb{Y}\had \left(\mb{O} + \mb{X}\mb{\Beta} + \mb{M}\right) - \mb{A} + \frac12 \log\mb{S} \bigg)  \mb{1}_{p}  \\
  + \frac{n}{2} \log \det{\mb{\Omega}} - \frac{n}{2}
  \tr\left(\hat{\boldsymbol\Sigma} \mb{\Omega}\right) + \frac{np}{2} -
  K(\mb{Y}),
\end{multline}
where $K(\mb{Y}) = \sum_{i,j} \log(Y_{ij}!)$ and $\had$ is the Hadamard 
(term-to-term) product.
\\

We now prove the biconcavity of $J$ and the same property will follow
for $J_\sp$. This result is the building block of the alternating
optimization algorithm that we propose in the upcoming section.
\begin{proposition}[Biconcavity of $J$] \label{prop:J_biconcave}
  $J$ is biconcave in $(\mb{\Beta}, \mb{M}, \mb{S})$ and
  $\mb{\Omega}$. Furthermore, if $\mb{X}$ has full rank, $J$ is
  strictly biconcave.
\end{proposition}

\begin{proof}
  We first prove the concavity of $J(\mb{\Beta}, \mb{M}, \mb{S})$.
  For fixed $\mb{\Omega}$, the quadratic form associated to the
  Hessian of $J$ is
\[ 
f: \bs{\theta} = \mtov(\Delta \mb{\Beta}, \Delta \mb{M}, \Delta \mb{S}) \mapsto f(\bs{\theta}) = \bs{\theta}^\trans \; \nabla^2_{\mb{\Beta}, \mb{M}, \mb{S}} J(\mb{\Beta}, \mb{M}, \mb{S}, \mb{\Omega}) \; \bs{\theta}. 
\]
Let $\sqrt{\mb{A}}$ be the
element-wise square-root of matrix $\mb{A}$ and $\mb{S}^\oslash$ the
element-wise inverse of matrix $\mb{S}$. The quadratic form simplifies to
\begin{align*}
 f(\bs{\theta}) & = - \tr([\sqrt{\mb{A}} \had \mb{X} \Delta \mb{\Beta}]^\trans[\sqrt{\mb{A}} \had \mb{X} \Delta \mb{\Beta}]) - 2\tr([\sqrt{\mb{A}} \had \mb{X} \Delta \mb{\Beta}]^\trans[\sqrt{\mb{A}} \had \Delta \mb{M}] ) \\ 
 &  -\tr([\sqrt{\mb{A}} \had \mb{X} \Delta \mb{\Beta}]^\trans[\sqrt{\mb{A}} \had \Delta \mb{S}] ) - \tr([\sqrt{\mb{A}} \had \Delta \mb{M}]^\trans[\sqrt{\mb{A}} \had \Delta \mb{M}] ) \\
 &  -\tr([\sqrt{\mb{A}} \had \Delta \mb{M}]^\trans[\sqrt{\mb{A}} \had \Delta \mb{S}] ) - \tr([\sqrt{\mb{A}} \had \Delta \mb{S} ]^\trans[\sqrt{\mb{A}} \had \Delta \mb{S}] )/4 \\ 
 &  - \tr(\Delta \mb{M} \mb{\Omega} \Delta \mb{M}^\trans) 
  - \tr([\mb{S}^{\oslash} \had \Delta \mb{S} ]^\trans[\mb{S}^{\oslash} \had \Delta \mb{S}] )/2 \\
 & = - \| \sqrt{\mb{A}} \had [\mb{X} \Delta \mb{\Beta} + \Delta \mb{M} + \Delta \mb{S}/2] \|_F^2 - \| \Delta\mb{M} \mb{\Omega}^{1/2}\|_F^2 - \| \mb{S}^{\oslash} \had \Delta\mb{S} \|_F^2 / 2\\
 & \leq 0,
\end{align*}
hence the Hessian matrix is negative semi-definite, which proves the
concavity of $J(\mb{\Beta}, \mb{M}, \mb{S})$. For strictness, consider
a triplet $(\Delta \mb{\Beta}, \Delta \mb{M}, \Delta \mb{S})$ such
that $ f(\bs{\theta}) = 0$. By definition of $\mb{S}^\oslash$ and the
positive definiteness of $\mb{\Omega}$,
$\Delta \mb{S} = \Delta \mb{M} = {0}$. Finally, since all entries in
$\mb{A}$ are positive, it leads to $\mb{X} \Delta \mb{\Beta} = {0}$
which implies $\Delta \mb{\Beta} = {0}$ as soon as ${\mb{X}}$ has full
rank. The lower bound $J(\mb{\Beta}, \mb{M}, \mb{S})$ is thus strictly concave 
with this assumption.

We now prove the concavity of $J(\mb{\Omega})$. The Hessian for fixed
$(\mb{\Beta}, \mb{M}, \mb{S})$ is
$$
-\frac{n}{2}\mb{\Omega}^{-1} \kro \mb{\Omega}^{-1},
$$
where $\kro$ denotes the Kronecker product. Since $\mb{\Omega}^{-1}$
is positive definite, so is $\mb{\Omega}^{-1} \kro \mb{\Omega}^{-1}$
and therefore $J$ is strictly concave in $\mb{\Omega}$.  .
\end{proof}

\begin{corollary}[Biconcavity of $J_\sp$] \label{prop:Biconcave}
  $J_\sp$ is biconcave in $(\mb{\Beta}, \mb{M}, \mb{S})$ and
  $\mb{\Omega}$. Furthermore, if $\mb{X}$ has full rank, $J_\sp$ is
  strictly biconcave.
\end{corollary}

\begin{proof} We use the concavity of
  $-\lambda\|\mb{\Omega}\|_{1,\off}$ and the fact that the sum of a
  strictly concave function with a concave function remains strictly
  concave.
\end{proof}

Unfortunately, $J$ (and consequently $J_\sp$) is not jointly convex in
$(\mb{\Beta}, \mb{M}, \mb{S}, \mb{\Omega})$ in general and
counter-examples can be found. In particular, this means that although
gradient descent will converge to a stationary point of $J$ (resp.
$J_\sp$), this stationary point is not guaranteed to be the global
optimum of $J$ (resp. $J_\sp$) and may depend on the starting point of
the iterative algorithm. Note that the same caveat applies
to alternating optimization schemes such as the (V)EM algorithm.

\paragraph{Alternate optimization.} To estimate both the variational
parameters $\bs{\psi}$ and the model parameter $\bs{\theta}$, we need
to maximize $J_\sp$ with the additional box constraint that
$\mb{S}\succ 0$, \textit{i.e.}, all variance parameters in the
variational distribution are strictly positive. We take advantage of
the biconcavity of $J_\sp$ and recourse to an alternating optimization
scheme to maximize $J_\sp$. At step $h$, the parameters are updated as
follows:
\begin{subequations}
  \label{eq:algorithm}
  \begin{align} \label{eq:algo_step1}
    (\mb{\Beta}^{(h)},\mb{M}^{(h)},\mb{S}^{(h)}) 
    & = \arg\max_{\mb{\Beta},\mb{M},\mb{S}\succ 0} J_\sp(\mb{Y}; (\mb{M},\mb{S}), (\mb{\Beta}, \mb{\Omega}^{h-1})) \nonumber \\
    & = \arg\max_{\mb{\Beta},\mb{M},\mb{S}\succ 0} J(\mb{Y}; (\mb{M},\mb{S}), (\mb{\Beta}, \mb{\Omega}^{h-1}))
  \end{align}
  \begin{equation}\label{eq:algo_step2}
    \mb{\Omega}^{(h)} = \arg\max_{\mb{\Omega}\in\mathbb{S}_{++}}
    J_\sp(\mb{Y}; (\mb{M}^{(h)},\mb{S}^{(h)}), (\mb{\Beta}^{(h)}, \mb{\Omega}))
  \end{equation}
\end{subequations}
where $\mathbb{S}_{++}$ is the set of positive-definite matrices.

Problem \eqref{eq:algo_step1} can be solved by a gradient ascent with
box-constraint for the variational variances $\mb{S}$ that must remain
nonnegative. We use the gradients which are given by
\begin{align}
  \label{eq:gradient}
  \begin{split}
   \nabla_{\mb{\Beta}} J & = \mb{X}^\trans (\mb{Y} - \mb{A}), \\
   \nabla_{\mb{M}} J & = \mb{Y} - \mb{A} - \mb{M} \mb{\Omega}, \\
   \nabla_{\mb{S}} J & = \frac12 \left(\mb{S}^\oslash - \mb{A} - \mb{1}_n \diag(\mb{\Omega})^\trans\right).\\
 \end{split}
\end{align}

When $\lambda>0$, Problem \eqref{eq:algo_step2} is easily shown to be
equivalent to solving 
\begin{equation}
  \label{eq:algo_step2_graphical_lasso}
  \minimize_{\mb{\Omega}\in\mathbb{S}_{++}} - \frac{n}{2} \log \det{\mb{\Omega}} + \frac{n}{2}
  \tr\left(\hat{\boldsymbol\Sigma} \mb{\Omega}\right) +  \lambda \; \|\mb{\Omega}\|_{\ell_1, \off}.
\end{equation}
We recognize a sparse multivariate Gaussian maximum likelihood problem
\citep{yuan2007,banerjee2008}, efficiently solved by the
graphical-Lasso algorithm \citep{FHT08}.

Finally, we alternate the two steps \eqref{eq:algo_step1} and
\eqref{eq:algo_step2} until convergence of the objective function
$J_\sp$. The algorithm is initialized using the estimator of the
graphical-Lasso obtained by shrinking the covariance matrix computed
on the Pearson residuals of a linear model predicting
$\log(1 + \mb{Y})$ from $\mb{X}$ and $\mb{O}$.

\paragraph{Model Selection.} Model selection is a notoriously hard
problem in unsupervised problems in general and in network inference
in particular. Several procedures have been proposed to select an
optimal value of $\lambda$ in Gaussian graphical models (GGM) and we
rely on both (i) the Stability Approach to Regularization Selection
(StARS) introduced in \citet{Liu2010} and (ii) variants of BIC
taylored for the high-dimensional setting, such as EBIC
\citep{chen2008extended}.

Briefly, StARS relies on resampling a large number $B$ of subsamples
of size $m$ (with or without replacement) and infers a network
$\mb{\Omega}^{(b, \lambda)}$ on each subsample $b$ for each value of
$\lambda$ in a grid $\Lambda$. The frequency of inclusion of edge
$e = i\sim j$ is computed as
$p_e^\lambda = \# \{b: \Omega^{(b, \lambda)}_{ij} \neq 0\}/B$ and its
variance as $v_e^\lambda = p_e^\lambda (1 - p_e^\lambda)$. The
stability $\stab(\lambda)$ of the network is then simply
$\stab(\lambda) = 1 - 2\bar{v}^\lambda$ where $\bar{v}^\lambda$ is the
average of the $v_e^\lambda$. Note that $\stab(\lambda)$ decreases
from $1$ for $\lambda = \infty$ (empty network) to a nonnegative value
for small $\lambda$. StARS selects the smallest $\lambda$ (densest
network) for which $\stab(\lambda) \geq 1 -
2\beta$.
\citeauthor{Liu2010} suggest using $2\beta = 0.05$ and subsamples of
size $m = \lfloor 10 \sqrt{n}\rfloor$ based on theoretical results. We
use them as default.

By contrast, BIC is a non-resampling based alternative with no
computational overhead. The extended family of BIC introduced in
\citet{chen2008extended} penalizes both the number of unknown
parameters and the complexity of the model space. In the framework of
PLNnetwork, we have the following expression
\begin{equation}
  \label{eq:EBIC}
  \text{EBIC}_\gamma(\hat{\mb{B}}, \hat{\mb{\Omega}}_\lambda)  =   -2 \, \textrm{loglik}
  (\mb{Y};\hat{\mb{B}}, \hat{\mb{\Omega}}_\lambda) + \log(n) (|\mathcal{E}_\lambda| + p d)
  + \gamma \log {p(p+1)/2 \choose |\mathcal{E}_\lambda|},
\end{equation}
where $\mathcal{E}_\lambda$ is the edge set of a candidate graph and
${m \choose n }$ corresponds to the binomial coefficient
(\textit{i.e.}, the number of models with $n$ parameters among $m$
possibles). The first penalty term in the right-hand-side is the usual
BIC penalization: our model has $pd$ unknown regression parameters in
$\mb{B}$ plus $|\mathcal{E}_\lambda|$ inferred terms in
$\hat{\mb{\Omega}}_\lambda$. The second penalty term, tuned by
$\gamma \in [0, 1]$, is used to adjust the tendency of the usual BIC
-- recovered for $\gamma = 0$ -- to choose overly dense graphs in the
high-dimensional setting. Here, we propose to replace
$\textrm{loglik}$ in \eqref{eq:EBIC} by its variational surrogate
\eqref{eq:var_bound}, that is, $J(\mb{Y};\hat{\mb{\Omega}})$ and use
$\gamma = 0$, \textit{i.e.}, simple BIC, instead of the value 0.5
recommended by \citeauthor{foygel2010extended} for GGM, that leads
almost systematically to empty networks in all our numerical
experiment.

\paragraph{Implementation.} We implemented our alternate optimization
algorithm in a \texttt{R/C++} package \citep{R} called
\textbf{PLNmodels}, available on github
\url{https://github.com/jchiquet/PLNmodels}. The Gradient ascent with
box constraints found in the first step is performed by means of the
implementation found in the \textbf{nlopt} library \citep{nlopt} of a
variant of the conservative convex separable approximation found in
\citep{svanberg2002class}. We use the \textbf{glasso} \texttt{R}
package \citep{FHT08} to solve the graphical-Lasso problem of the
second step.


\section{Simulation study} \label{sec:simul}

\subsection{Simulation protocol}

\paragraph*{Network generation.} The ground truth graphs that
originate the precision matrices are generated according to various
random graph-models, namely Erd\"os-R\'enyi model (no particular
structure), preferential attachment model (scale-free property) or
affiliation model (community structure). These models are used to
generate a binary adjacency matrix $\mb{G}$ from which we build a
precision matrix $\mb{\Omega}$ that must be positive-definite while
sharing the same sparsity pattern (but for the diagonal) as $\mb{G}$. We
ensure these two properties as follows:
\begin{equation*}
  \tilde{\mb{\Omega}} = \mb{G} \times v, \quad \mb{\Omega} = \tilde{\mb{\Omega}} + \diag(|\min(\mathbf{eig}(\tilde{\mb{\Omega}}))| + u ), \quad \text{with } u,v > 0.
\end{equation*}
The two scalars $u,v$ are used to partially control the difficulty of
the network inference problem: they are related to the strength of the
partial correlations -- and in turn of the interactions in the network
-- while they also control the conditioning of $\mb{\Omega}$. Higher
$v$ leads to stronger correlations and higher $u$ to better
conditioning. We always set $v=0.3, u=0.1$ in our simulations. This
protocol is similar to the one at play in the \texttt{R} package \textbf{huge}.

\paragraph*{Compositional data generation.} In order not to promote
any network reconstruction method in particular and thus provide fair
comparisons, the simulated count data are not drawn according to a PLN
distribution. Instead, we introduce a compositional model inspired
from community ecology data. This model also applies to sequencing
data in genomics where counts are not comparable between samples,
since sequencing technologies do not provide an absolute measurements
of species or gene abundances. We sketch the process of data
generation in Figure~\ref{fig:compositional_data}, the steps of which
are:
\begin{enumerate}[i)]
\item Draw the 'real' (unreachable) \emph{abundances} $\mathbf{a}_i$
  of the $p$ species in sample $i$ such that
  $\log(\mathbf{a}_i) \sim \mathcal{N}(\mathbf{\mb{X} \mb{B}},
  \mb{\Omega}^{-1})$; the design matrix $\mb{X}$ accounts for some
  covariates and $\mb{\Omega}$ is the latent network between species
  drawn as explained above.
\item Transform abundances $\mathbf{a}_i$ to \emph{proportions}
  $\boldsymbol\pi_{i}$ with logistic-transform, i.e. $\pi_{ij} = e^{b_{ij}}/\sum_j e^{b_{ij}}$.
\item For random value of $N_i$ -- the sampling effort in sample $i$,
  typically the sequencing depth -- draw observed \emph{counts} $Y_i$
  via a multinomial distribution
  $\mathcal{M}(N_i, \boldsymbol\pi_i)$.
\end{enumerate}

\begin{figure}[htbp!]
  \flushright
  \begin{minipage}[c]{.9\textwidth}
    \begin{scriptsize}
      \begin{minipage}[c]{0.35\textwidth}            
          \begin{minipage}[c]{.45\textwidth}
            network $\rightsquigarrow$ $\mb{\Omega}$  \hfill \\
            \begin{tikzpicture}[scale=.8]
              \node[shape=circle, draw=black] (A) at (-1,0) {A};
              \node[shape=circle, draw=black] (B) at (0,0) {B};
              \node[shape=circle, draw=black] (C) at (0,-1) {C};
              \node[shape=circle, draw=black] (D) at (1,0.5) {D};
              \node[shape=circle, draw=black] (E) at (1,-0.5) {E};
              \path[-][very thick, black!50!green] (A) edge node[left] {} (B);
              \path[-][very thick, black!50!green] (B) edge node[left] {} (C);
              \path[-][very thick, black!50!green] (B) edge node[left] {} (D);
              \path[-][very thick, black!30!red] (C) edge node[left] {} (E);
              \draw[dashed,->,very thick] (1.5,1) -- (2,1.5);
              \draw[dashed,->,very thick] (1.5,-1) -- (2,-1.5);
            \end{tikzpicture}
          \end{minipage}
          \hfill
          \begin{minipage}[c]{0.45\textwidth}
            abundances~$\log\mathbf{a}_i$
            \[\left(
                {\begin{array}{c}
                   8 \\
                   8 \\
                   8 \\
                   9 \\
                   2
                 \end{array}
               }\right)\]
           \[\left(
               {\begin{array}{c}
                  6 \\
                  6 \\
                  6 \\
                  8 \\
                  1
                \end{array}
              }\right)\]
        \end{minipage}      
    \end{minipage}
    \begin{minipage}[c]{0.2\textwidth}    
      \begin{center}
        \begin{minipage}[c]{0.355\linewidth}
          \vspace*{0.7cm}
          \begin{tikzpicture}[scale=1]
            \draw[->,very thick] (0,1) -- (0.6,1);
            \draw[->,very thick] (0,-2) -- (0.6,-2);
          \end{tikzpicture}
        \end{minipage} \hfill
        \begin{minipage}[c]{.6\linewidth}
          proportions~${\boldsymbol\pi}_i$
          \[\left(
              {\begin{array}{c}
                 0.17 \\
                 0.17 \\
                 0.17 \\
                 0.47 \\
                 0.02
               \end{array}
             }\right)\]
         \[\left(
             {\begin{array}{c}
                0.09 \\
                0.09 \\
                0.09 \\
                0.71 \\
                0.02
              \end{array}
            }\right)\]
      \end{minipage}
    \end{center}
  \end{minipage}
  \begin{minipage}[c]{0.3\textwidth}
    \begin{center}
      \begin{minipage}[c]{0.20\textwidth}
        \vspace*{0.7cm}
        \begin{tikzpicture}[scale=1]
          \draw[->,very thick] (0,0) -- (0.6,0);
        \end{tikzpicture}
      \end{minipage}
      \hfill
      \begin{minipage}[c]{0.75\textwidth}
        counts
        \begin{tabular}{c|c|c|}
          \multicolumn{1}{c}{~} & \multicolumn{1}{c}{$i = 1$} & \multicolumn{1}{c}{$i = 2$} \\
          \hhline{~|-|-|}
          A &   3        &    8        \\ \hhline{~|-|-|}
          B &   3        &    6        \\ \hhline{~|-|-|}
          C &   5        &    2        \\ \hhline{~|-|-|}
          D &   8        &    30       \\ \hhline{~|-|-|}
          E &   1        &    4        \\ \hhline{~|-|-|}
          \multicolumn{1}{c}{~}        \\ \hhline{~|-|-|}
          $N_i$ & 20         & 50          \\ \hhline{~|-|-|}
        \end{tabular}
      \end{minipage}
    \end{center}
  \end{minipage}
\end{scriptsize}
\end{minipage}

\caption{Compositional model used for data generation}
\label{fig:compositional_data}
\end{figure}
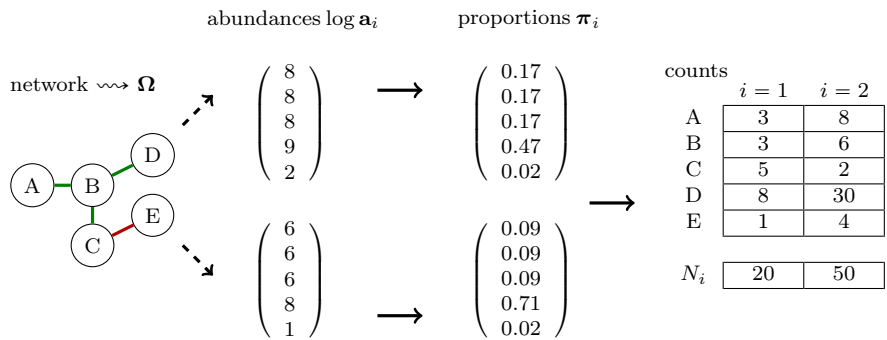

\paragraph*{Experimental setup.} We fix the number of variables to
$p = 50$ in all our experiments. Indeed, the networks with a number of
nodes of this order of magnitude are the largest ones that can be
decently analyzed by biologists in genomics or ecology given the number of 
samples at hand. They also correspond to the (order of magnitude of) the number 
of nodes considered in Section \ref{sec:appli} for the three real-world
applications.

The sampling effort $N_i$ are drawn from a negative binomial
distribution so that
$N_i\sim^{\text{i.i.d}} \mathcal{NB}(\mu=1000, \nu)$, that is, a
mean total count number of $1000$ per sample with a variance equal to
$1000 + 1000^2/\nu$. The covariates are chosen so that $\mb{X}$ is
the design matrix of a one-way ANOVA with 3 balanced groups, hence
$d= 3$. The regression coefficients are sampled in a uniform
distribution so that $\Beta_{jk}\sim^{\text{i.i.d.}}\mathcal{U}(-b,b)$. Those 
parameters were chosen to replicate (marginal) count distributions 
-- in terms of location and dispersion -- commonly observed in 
microbial ecology applications. 

To control the difficulty of the problem, we vary the sample size $n$
as well as the following quantities:
\begin{enumerate}[i)]
\item the overdispersion of the sampling efforts $N_i$: the larger
  $\nu$, the smaller the overdispersion and the more similar the
  samples;
\item the effect of the covariates $\mb{X}\mb{\Beta}$: the larger $b$,
  the larger the Signal to Noise Ratio (SNR) in the underlying linear model 
and the smaller the fraction of variance explained by $\mb{\Omega}$.
\end{enumerate}

\paragraph*{Competitors.} For all numerical experiments and simulations, we 
refer to the implementation of a given competitor as
its name using \texttt{teletype} family font. For instance, our method
is referred to as \texttt{PLNnetwork}.

Among the many possible competitors to \texttt{PLNnetwork}, we pick
some representatives dispatched in the three following families of
method:
\begin{enumerate}
\item Vanilla sparse GGM methods \citep{FHT08,Meinshausen2006} applied
  after a log-transformation of the count. We choose the
  \texttt{graphical-Lasso} as implemented in the \texttt{R}-package
  \textbf{glasso} \citep{FHT08}, with log transformation of the
  count table as pretreatment.
\item Sparse log-linear graphical models
  \citep{yang2012graphical,allen2012log}, referred to as
  \texttt{sparse LLM} in the following. We rely on the implementation
  found in the \texttt{R}-package \textbf{RNAseqNet}
  \citep{imbert2017multiple}.
\item Methods dedicated to compositional count data, whose
  gold-standard approaches are \texttt{SPiEC-Easi} \citep{Kurtz2015}
  for the precision matrix or \texttt{sparCC} \citep{Friedman2012} for the 
correlation one. Both methods account for
  compositional data by using pseudo-counts plus
  log-transformation. The former applies graphical-Lasso and
  non-paranormal transformation \citep{liu2009nonparanormal}. The
  latter uses resampling and thresholded correlations. The
  \texttt{R}-package \textbf{spieceasi} provides an implementation of
  these two methods
  .
\end{enumerate}

\paragraph*{Performance assessment.} Each competitor produces a
sequence of inferred networks indexed by a tuning parameter that
controls the number of edges in the final estimator, from an empty to
a full graph, ordered by reliability. Since the problem of choosing
tuning parameters is known as particularly troublesome in unsupervised
problems like network inference, the reconstruction methods are
commonly compared by means of precision-recall (PR) or Receiver
operating characteristic (ROC) curves that leave the choice of a
particular tuning parameter aside. We recall that ROC curves are
obtained by plotting the true positive rate (or recall) as a function
of the false positive rate (or fall-out), while PR curve represents
the positive predictive value (or precision) as a function of the
recall. While the former is more spread in the literature, the latter
is more informative in unbalanced cases with a small proportion of
positives. Indeed, PR gives less weight to regions with a large false
positive rate, which are generally not interesting for the
practitioners \citep{davis2006relationship}. We use both of them in
our experiments, and use area under the ROC curve (AUC) and area under
the PR curve (AUPR) to summarize one simulation: the closer to one,
the better the network reconstruction.

\subsection{Results}

We now present the results of two batches of numerical experiments
that illustrate the effect of different experimental factors (namely
the sampling effort and the presence of an external covariate) on the
quality of the network reconstruction. On top of these experiments, we
present a numerical study that address the model selection issue in
\texttt{PLNnetwork}, that is, the choice of the tuning parameter
$\lambda$.

\paragraph*{Non-compositional methods fail.} We first study the effect
of a different sampling effort between the sample on the quality of
the network reconstruction by varying the value of
$\nu \in \{100, 10, 2\}$ (corresponding to a small, medium and a large
variability) in the compositional model. We compare
\texttt{graphical-Lasso}, \texttt{sparse LLM} and \texttt{PLNnetwork},
the latter being the only method accounting for the compositional
problem, by introducing an offset which is sample dependent. This
offset is computed as the total sum of counts found in each
sample. Results averaged over 100 replicates are displayed in
Figure~\ref{fig:compositional_problem}.  The first row shows the AUC
for varying sample size and a different variability between
samples. As expected, \texttt{PLNnetwork} is the only method which is
not sensitive to the sampling effort, contrary to
\texttt{graphical-Lasso} and \texttt{sparse LLM} which completely fail
at recovering the dependence structure in presence of some unaccounted
source of variability between the samples. In the second row, the AUPR
exhibits an even larger discrepancies between the compositional and
non-compositional methods: while the AUC is close to 1 for a small
effect of the variability and a large sample size, the AUPR attests
that the first edges inferred by \texttt{graphical-Lasso} and
\texttt{sparse LLM} are in fact most of the time false positives.
\begin{figure}[htbp!]
  \centering
  \begin{tabular}{@{}c@{}}
    \hspace{1.5cm} \small Variance of the sampling effort \\
   \includegraphics[width=\textwidth]{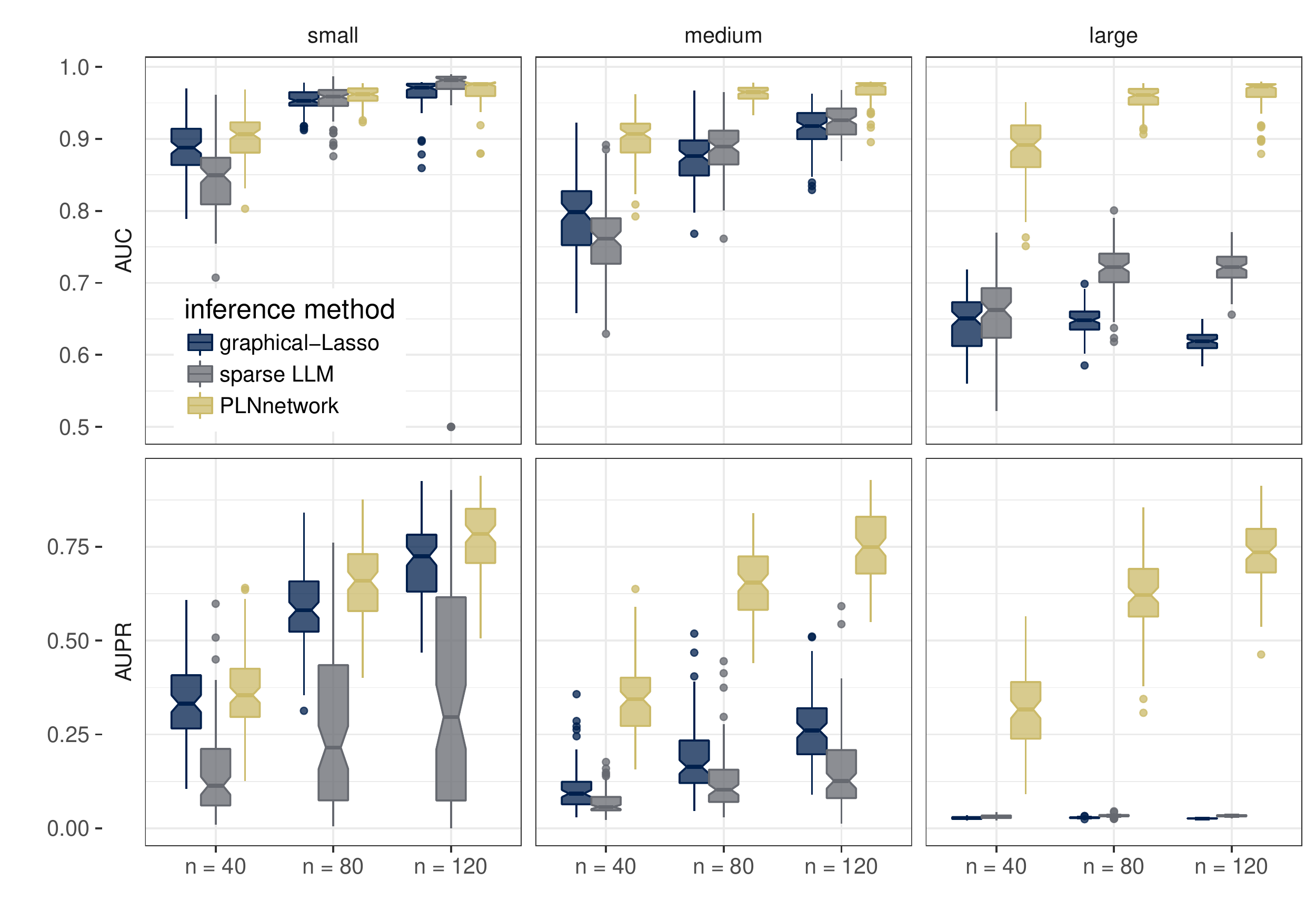}
  \end{tabular}
  \caption{Effect of the variability of the sampling effort, a.k.a the
    compositional problem, on the quality of the reconstruction of
    50-node random networks. First and second row respectively
    represent AUC and AUPR boxplots computed on 100 simulations.  }
  \label{fig:compositional_problem}
\end{figure}

\paragraph*{Accounting for covariates effect does matter.} We now
focus on the effect of an external covariate in the data, and how it
affects the performance of the methods. Regarding the sampling effort,
we fix $\nu = 2$ in this experiment, and we only compare the
compositional methods together since the other approaches would fail
in this setting. The strength of the covariate effect is controlled by
the parameter $b$ in our compositional model. The larger $b$, the
larger the effect of the covariate and the harder the problem of
network reconstruction when not accounting for the covariate. 
We vary $b\in\{1,2,3\}$, hence, a small, medium and large effect. On top of 
that, we vary the sample size and consider the three network topologies 
(scale-free, random and community networks), always with $p=50$ nodes. We 
evaluate the performance of \texttt{SPiEC-Easi}, \texttt{sparCC} and
\texttt{PLNnetwork} in terms of AUC and AUPR on 100 simulation of each
kind and report the average values in Table~\ref{tab:res_inference}.
\begin{table}[htbp!]
  \centering
  \begin{small}
    
\begin{tabular}{@{}l@{\hspace{\tabcolsep}}l@{\hspace{\tabcolsep}}r@{ (}r@{)\hspace{\tabcolsep}}r@{ (}r@{)\hspace{\tabcolsep}}r@{ (}r@{)\hspace{2\tabcolsep}}r@{ (}r@{)\hspace{\tabcolsep}}r@{ (}r@{)\hspace{\tabcolsep}}r@{ (}r@{)}}
  \hline  
   & & \multicolumn{6}{c}{\textbf{area under the ROC}} & \multicolumn{6}{c}{\textbf{area under the PR}} \\[-1ex]
   & & \multicolumn{6}{c@{}}{\bf\hrulefill} & \multicolumn{6}{c@{}}{\bf\hrulefill} \\
   \textbf{covar.} & \textbf{method} &
   \multicolumn{2}{c}{$\mathbf{n=p/2}$} & \multicolumn{2}{c}{$\mathbf{n=p}$} & \multicolumn{2}{c}{$\mathbf{n=2p}$} & 
   \multicolumn{2}{c}{$\mathbf{n=p/2}$} & \multicolumn{2}{c}{$\mathbf{n=p}$} & \multicolumn{2}{c}{$\mathbf{n=2p}$} \\[.5ex] 
  \hline
  \multicolumn{14}{c}{\textbf{scale-free network}} \\
  \hline
  small  & \texttt{PLNnetwork} & .66          & 0.05 & \textbf{.78} & 0.05 & \textbf{.91} & 0.03 & \textbf{.11} & 0.04 & \textbf{.25} & 0.07 & \textbf{.49} & 0.08 \\  
         & \texttt{sparCC}     & .66          & 0.05 & .73          & 0.05 & .79          & 0.05 & .09          & 0.03 & .16          & 0.05 & .24          & 0.07 \\  
         & \texttt{SPiEC-Easi} & \textbf{.67} & 0.04 & .77          & 0.05 & .85          & 0.04 & .10          & 0.03 & .17          & 0.05 & .27          & 0.07 \\[1ex]  
  medium & \texttt{PLNnetwork} & \textbf{.62} & 0.05 & \textbf{.73} & 0.05 & \textbf{.85} & 0.05 & \textbf{.09} & 0.03 & \textbf{.18} & 0.06 & \textbf{.34} & 0.08 \\  
         & \texttt{sparCC}     & .55          & 0.05 & .57          & 0.05 & .58          & 0.05 & .05          & 0.01 & .05          & 0.01 & .06          & 0.01 \\  
         & \texttt{SPiEC-Easi} & .61          & 0.04 & .66          & 0.04 & .71          & 0.03 & .06          & 0.01 & .06          & 0.01 & .07          & 0.01 \\[1ex]  
  large  & \texttt{PLNnetwork} & \textbf{.58} & 0.05 & \textbf{.67} & 0.05 & \textbf{.78} & 0.05 & \textbf{.07} & 0.03 & \textbf{.12} & 0.04 & \textbf{.23} & 0.07 \\  
         & \texttt{sparCC}     & .52          & 0.04 & .53          & 0.04 & .53          & 0.05 & .04          & 0.01 & .04          & 0.01 & .04          & 0.01 \\  
         & \texttt{SPiEC-Easi} & .57          & 0.04 & .60          & 0.03 & .65          & 0.03 & .05          & 0.01 & .05          & 0.01 & .05          & 0.01 \\ 
  \hline
  \multicolumn{14}{c}{\textbf{random network}} \\
  \hline
  small  & \texttt{PLNnetwork} & .77          & 0.07 & \textbf{.90} & 0.04 & \textbf{.96} & 0.01 & \textbf{.14} & 0.07 & \textbf{.36} & 0.11 & \textbf{.64} & 0.09 \\  
         & \texttt{sparCC}     & .76          & 0.06 & .83          & 0.06 & .89          & 0.04 & .11          & 0.05 & .23          & 0.09 & .36          & 0.11 \\  
         & \texttt{SPiEC-Easi} & \textbf{.78} & 0.05 & .87          & 0.04 & .92          & 0.03 & .11          & 0.05 & .23          & 0.09 & .36          & 0.11 \\[1ex]  
  medium & \texttt{PLNnetwork} & \textbf{.72} & 0.06 & \textbf{.85} & 0.05 & \textbf{.94} & 0.02 & \textbf{.09} & 0.04 & \textbf{.24} & 0.09 & \textbf{.49} & 0.10 \\  
         & \texttt{sparCC}     & .59          & 0.06 & .61          & 0.07 & .62          & 0.06 & .03          & 0.01 & .04          & 0.02 & .04          & 0.02 \\  
         & \texttt{SPiEC-Easi} & .67          & 0.05 & .74          & 0.05 & .77          & 0.03 & .04          & 0.01 & .05          & 0.02 & .05          & 0.01 \\[1ex]  
  large  & \texttt{PLNnetwork} & \textbf{.64} & 0.07 & \textbf{.78} & 0.06 & \textbf{.88} & 0.04 & \textbf{.06} & 0.03 & \textbf{.14} & 0.07 & \textbf{.29} & 0.09 \\  
         & \texttt{sparCC}     & .54          & 0.05 & .53          & 0.06 & .54          & 0.06 & .02          & 0.01 & .02          & 0.01 & .03          & 0.01 \\  
         & \texttt{SPiEC-Easi} & .61          & 0.05 & .65          & 0.04 & .68          & 0.03 & .03          & 0.00 & .03          & 0.00 & .03          & 0.01 \\ 
  \hline
  \multicolumn{14}{c}{\textbf{community network}} \\
  \hline
  small  & \texttt{PLNnetwork} & .60          & 0.04 & .69          & 0.04 & \textbf{.78} & 0.05 & \textbf{.17} & 0.03 & \textbf{.26} & 0.04 & \textbf{.38} & 0.05 \\  
         & \texttt{sparCC}     & \textbf{.62} & 0.04 & .66          & 0.04 & .70          & 0.04 & .16          & 0.02 & .21          & 0.04 & .26          & 0.04 \\  
         & \texttt{SPiEC-Easi} & \textbf{.62} & 0.04 & \textbf{.70} & 0.04 & .77          & 0.04 & \textbf{.17} & 0.02 & .24          & 0.04 & .31          & 0.04 \\[1ex]  
  medium & \texttt{PLNnetwork} & .57          & 0.03 & \textbf{.65} & 0.04 & \textbf{.73} & 0.05 & \textbf{.15} & 0.02 & \textbf{.22} & 0.03 & \textbf{.31} & 0.05 \\  
         & \texttt{sparCC}     & .55          & 0.03 & .56          & 0.04 & .56          & 0.03 & .11          & 0.02 & .12          & 0.02 & .12          & 0.02 \\  
         & \texttt{SPiEC-Easi} & \textbf{.58} & 0.03 & .63          & 0.03 & .67          & 0.03 & .13          & 0.02 & .14          & 0.02 & .15          & 0.02 \\[1ex]  
  large  & \texttt{PLNnetwork} & \textbf{.55} & 0.03 & \textbf{.60} & 0.04 & \textbf{.67} & 0.04 & \textbf{.13} & 0.02 & \textbf{.17} & 0.03 & \textbf{.24} & 0.04 \\  
         & \texttt{sparCC}     & .52          & 0.03 & .52          & 0.03 & .52          & 0.03 & .10          & 0.02 & .10          & 0.02 & .10          & 0.02 \\  
         & \texttt{SPiEC-Easi} & \textbf{.55} & 0.03 & .58          & 0.03 & .62          & 0.03 & .11          & 0.01 & .11          & 0.02 & .12          & 0.01 \\  
  \hline
\end{tabular}
  \end{small}
  \caption{Areas under the ROC curve and Areas under the
    Precision-Recall curve of the compositional methods
    (\texttt{PLNnetwork}, \texttt{sparCC} and \texttt{SPiEC-Easi}) in
    various settings, averaged over 100 simulations, with standard
    errors.}
\label{tab:res_inference}
\end{table}

\texttt{PLNnetwork}, which is the only method that can effectively
account for the covariate effect, is a clear leader in almost all
settings. Even when the effect of the covariate is small, it seems to
outperform its competitors, except in the high dimensional setup and
for community network, where all methods perform similarly. We underline that 
the superiority of \texttt{PLNnetwork} is even clearer in terms of AUPR. In 
other words, \texttt{PLNnetwork} --
as a statistical method -- has a similar or higher power than its
competitors while keeping a lower false positive rate.

\paragraph*{Model Selection issue.}  In this numerical study, we focus
on \texttt{PLNnetwork} and address the difficult question of choosing
the tuning parameter and compare the two alternatives presented in
Section~\ref{sec:inference}. Figure~\ref{fig:selection} reports the
results of our numerical experiments on model selection: we compare
the StARS criterion computed on 50 subsamples with a stability
threshold of $1-2\beta = 0.95$ as recommended by
\citeauthor{liu2010stability}, to the BIC for choosing $\lambda$ in
\texttt{PLNnetwork}. We give the performance in terms of
precision/recall and fall-out/recall averaged over 100 simulation,
which correspond to single points on the PR and ROC curves
respectively. As can been seen, StARS systematically outperforms BIC
in terms of recall and precision. 
However, this increase in performance comes at a huge
computational burden.

\begin{figure}[htbp!]
  \centering
  \begin{tabular}{lc@{\hspace{1cm}}lc}
    & \hspace{.5cm} \small choice on ROC curve &  & \hspace{.5cm} \small choice 
on PR curve \\
     \rotatebox{90}{\hspace{2.7cm}\small recall} & 
\includegraphics[width=.4\textwidth]{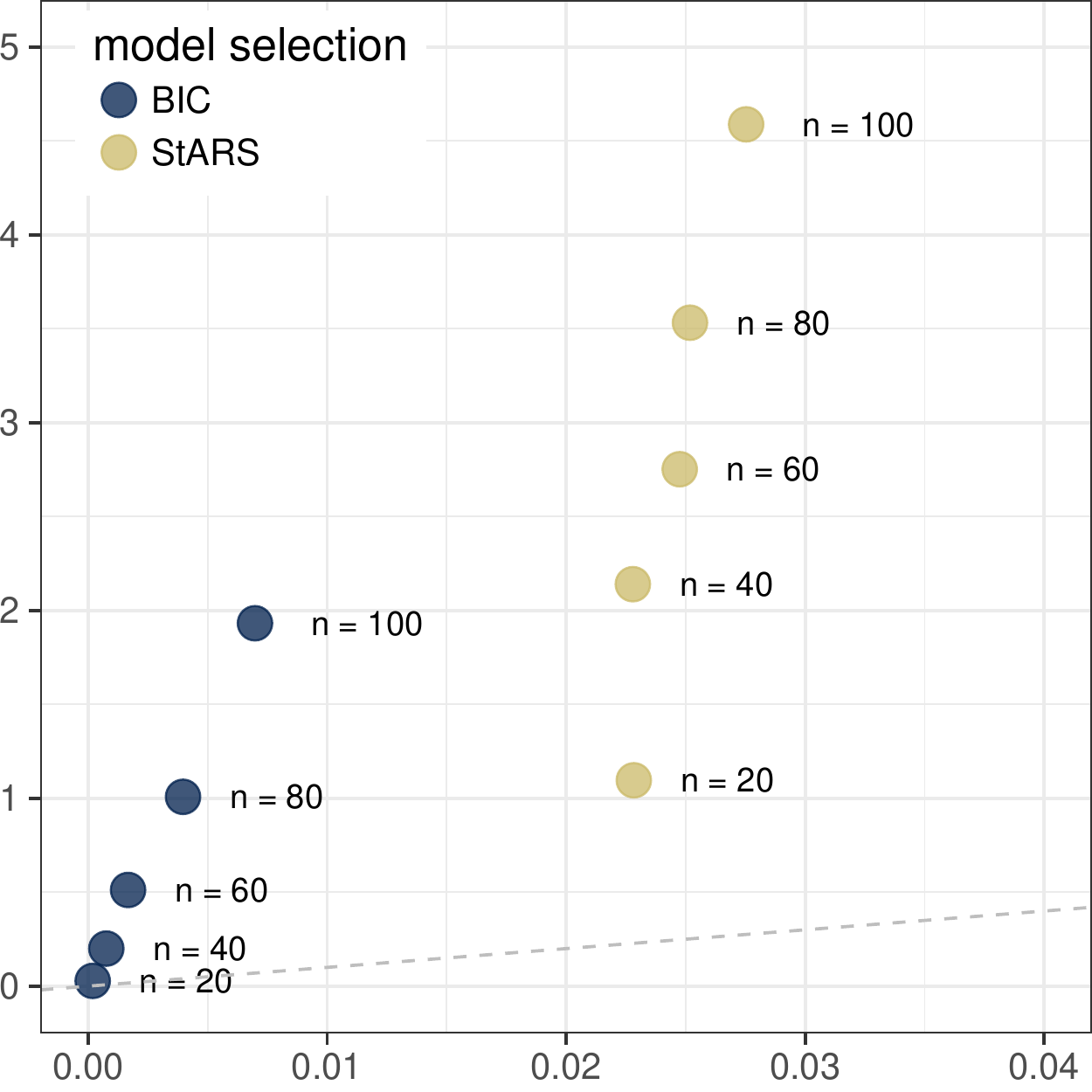}
   & \rotatebox{90}{\hspace{2.7cm}\small precision} & 
\includegraphics[width=.4\textwidth]{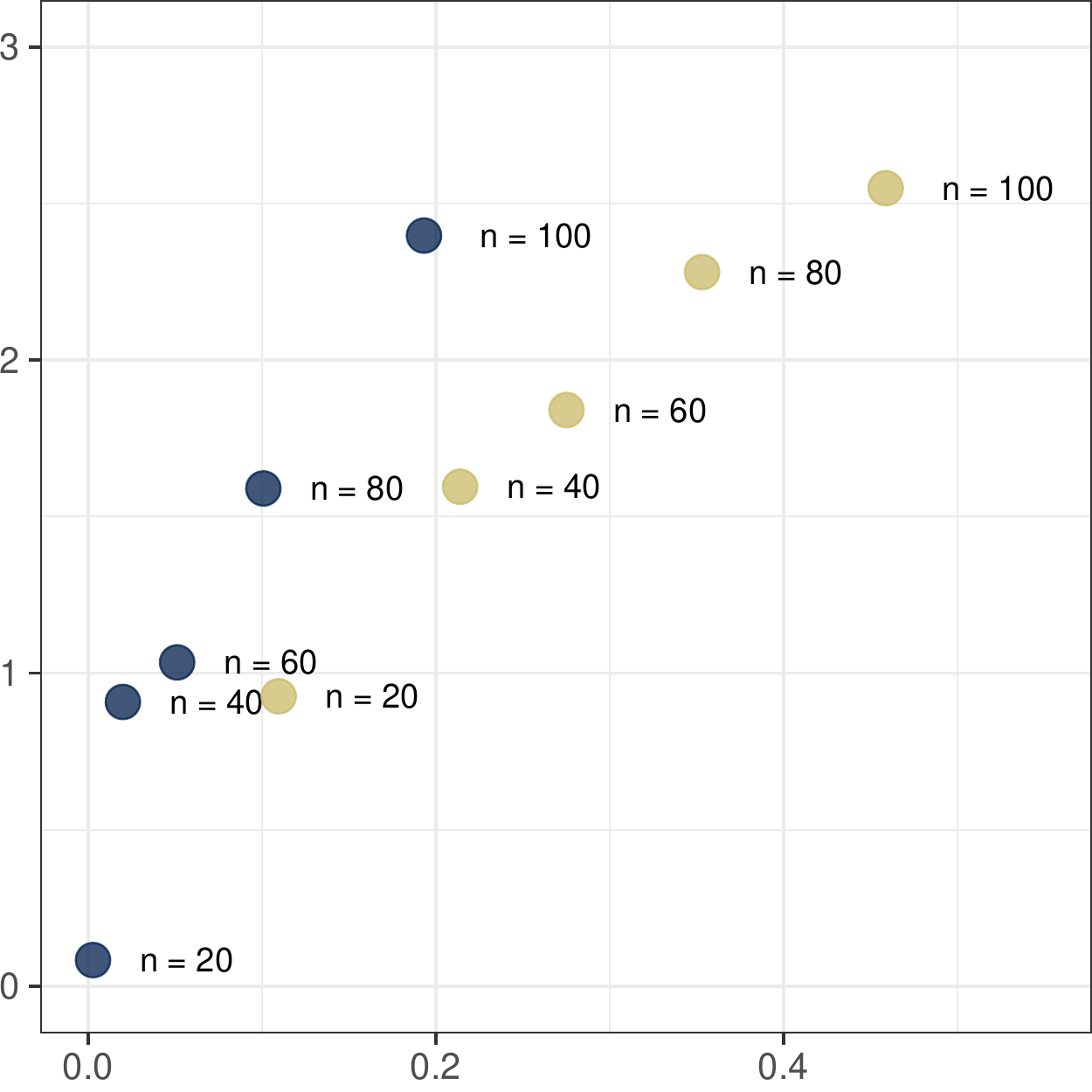} \\
   & fall-out & & recall\\
  \end{tabular}
  \caption{Performance of the model selection procedures (BIC or
    StARS) in \texttt{PLNnetwork} for reconstructing 50-node randoms
    networks, averaged over 100 simulations.}
  \label{fig:selection}
\end{figure}


\section{Illustrations} \label{sec:appli}

We now illustrate our methodology with a series of examples from different fields.
Barents fish is a simple ecological example that we use to emphasizes the importance of accounting for covariates when performing structure inference. The French election example shows that our method can handle large datasets; we also use it to show how to interpret the results and propose some validation checks. Finally we consider a metagenomic example (Oak mildew), for which we propose a deeper analysis. More specifically, we show how to decompose the effects of the different covariates on the inferred interactions and we propose some biological interpretations of the results.

\subsection{Barents fish} \label{sec:barents}

The data consist in the abundance of $p=30$ fish species measured in $n=89$ 
stations from the Barents sea between April and May 1997. The data have been 
collected and described by \cite{FNA06} and re-analyzed by \cite{Gre13,GrP14}. 
For each sample, the latitude and longitude of the station as well as the 
temperature and depth were recorded. Thanks to a precise experimental protocol, 
all abundances are comparable so no offset term is required in the model. Our 
aim here is to illustrate how the inclusion of covariates avoids spurious edges 
in the inferred network.

\paragraph{Introducing covariates reduces the number of inferred edges.}
To this aim, we fitted the \PLNnetwork model with ($a$) no covariates, ($b$) two 
environmental covariates (temperature and depth) and ($c$) all covariates 
(\emph{i.e.} the previous two and geographical location) using the same 
penalty grid every time with $\lambda$ increasing geometrically from $0.03$ to 
$15.17$. As expected, Figure \ref{fig:barents-networks} (top right panel) shows 
that, for all models, the number of edges increases as the penalty decreases. It 
also shows that, for any penalty, the number of edges decreases as (plain lines) 
the richness of the model increases ($c > b > a$) and that most edges recovered 
in the full ($c$) model are also recovered in the partial models ($a$, black dotted 
curve) and ($b$, blue dotted curve). This suggests that naive inference is likely 
to find not only genuine edges but also spurious ones corresponding to 
co-variations induced by external covariates. Interestingly, the dotted 
curve shows that the proportion of common edges between models $b$ and $c$ is 
higher than the one between models $a$ and $c$. This suggests that 
environmental covariates explain a substantial part of the apparent 
species co-variations.

\paragraph{Spurious interactions can be linked with specific covariates.}
The rest of Figure \ref{fig:barents-networks} displays the networks
inferred with the three models for three different levels of sparsity
(controlled by $\lambda$). For an illustrative purpose, the values of
$\lambda$ have been chosen so that, in average, each species interacts
with two others for each of the three models ($a$), ($b$) and
($c$). This results in networks with approximately $2p = 60$
edges. The comparison of these networks confirms the conclusions
obtained in the simulation study. One additional conclusion is that a
set of core species seem to have direct interactions, or at least,
interactions that cannot be simply explained by geographical location
and environmental covariates (bottom right panel).

On the contrary, some interactions seem to be actually indirect. For example, 
the interactions between the longear eelpout ({\sl Ly.se}) and some species from 
the core group disappear when accounting for temperature and depth, 
suggesting that the covariation of their respective abundances results from 
variations of the environmental conditions. Similarly, the interactions 
between the Greenland halibut ({\sl Re.hi}) and the core group is kept when 
introducing temperature and depth in the model, but disappears when accounting 
for location (longitude and latitude) suggesting that these interactions 
actually reflect a common response to fluctuations of biotic and abiotic 
characteristics across sites. 

To confirm this interpretation, we fitted an over-dispersed Poisson generalized 
linear model for the abundance of both species (not shown). We found that both 
temperature and depth have a significant effect on the abundance of the 
longear eelpout and that the longitude has a significant influence on the 
abundance of the Greenland halibut (all corresponding $p$-values being smaller 
than $1e^{-4}$).

\begin{figure}
 \begin{center}
  \begin{tabular}{cccc}
    & ($a$) no covariate & ($b$) temperature \& depth & ($c$) all covariates \\
    \hline
    \rotatebox{90}{\hspace{1.8cm} $\lambda=.84$} &
    \includegraphics[width=.32\textwidth]{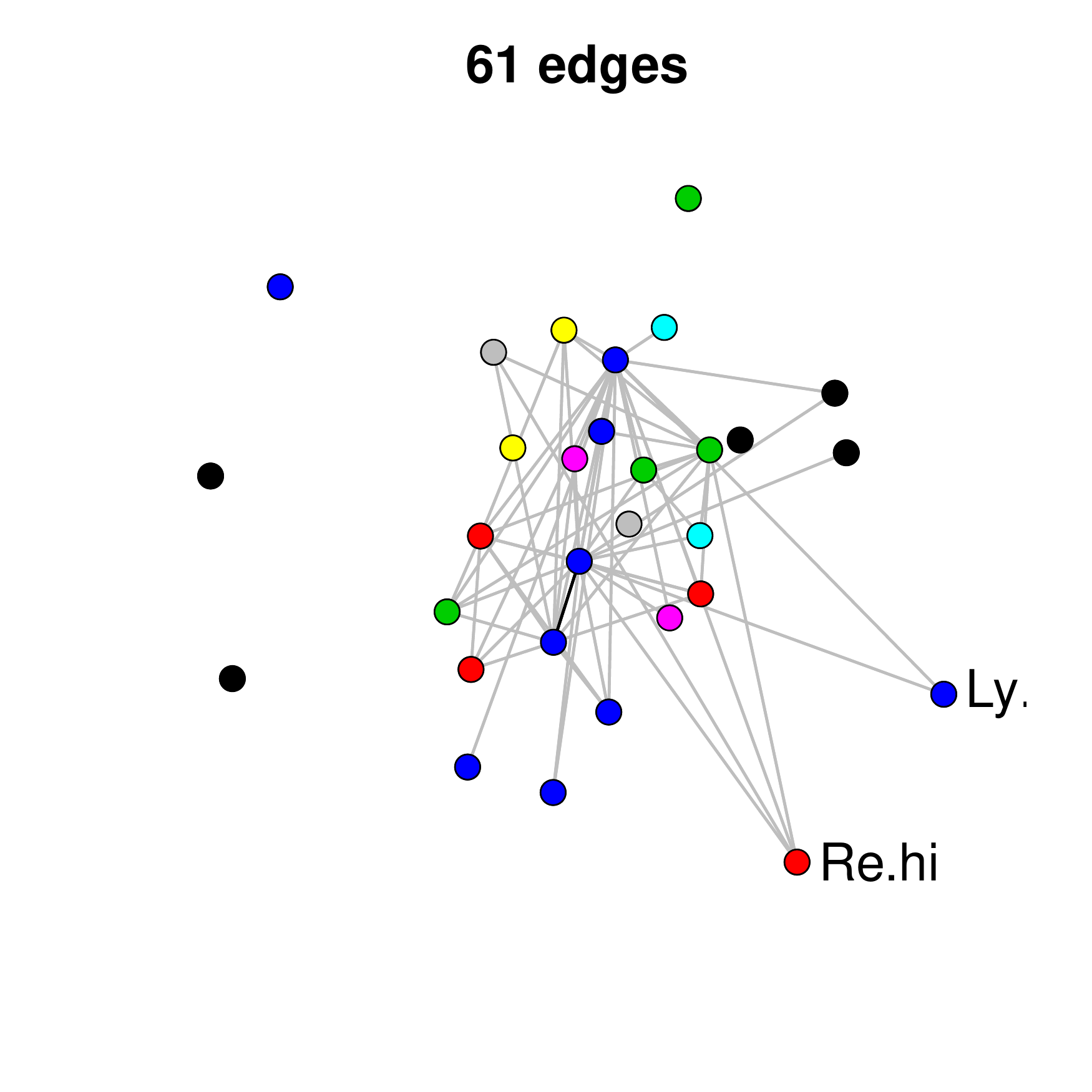} &
    \includegraphics[width=.32\textwidth]{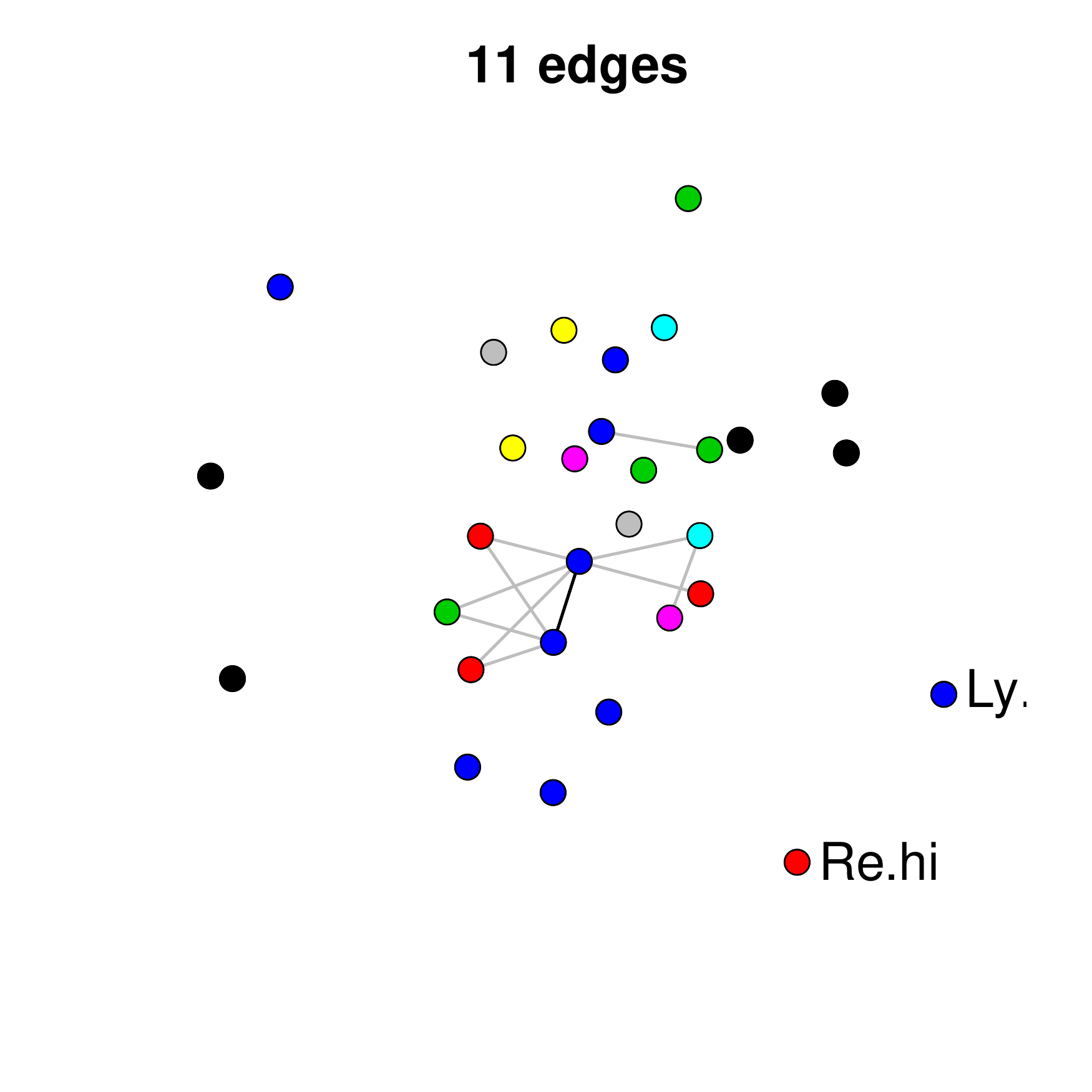} &
    \includegraphics[width=.32\textwidth]{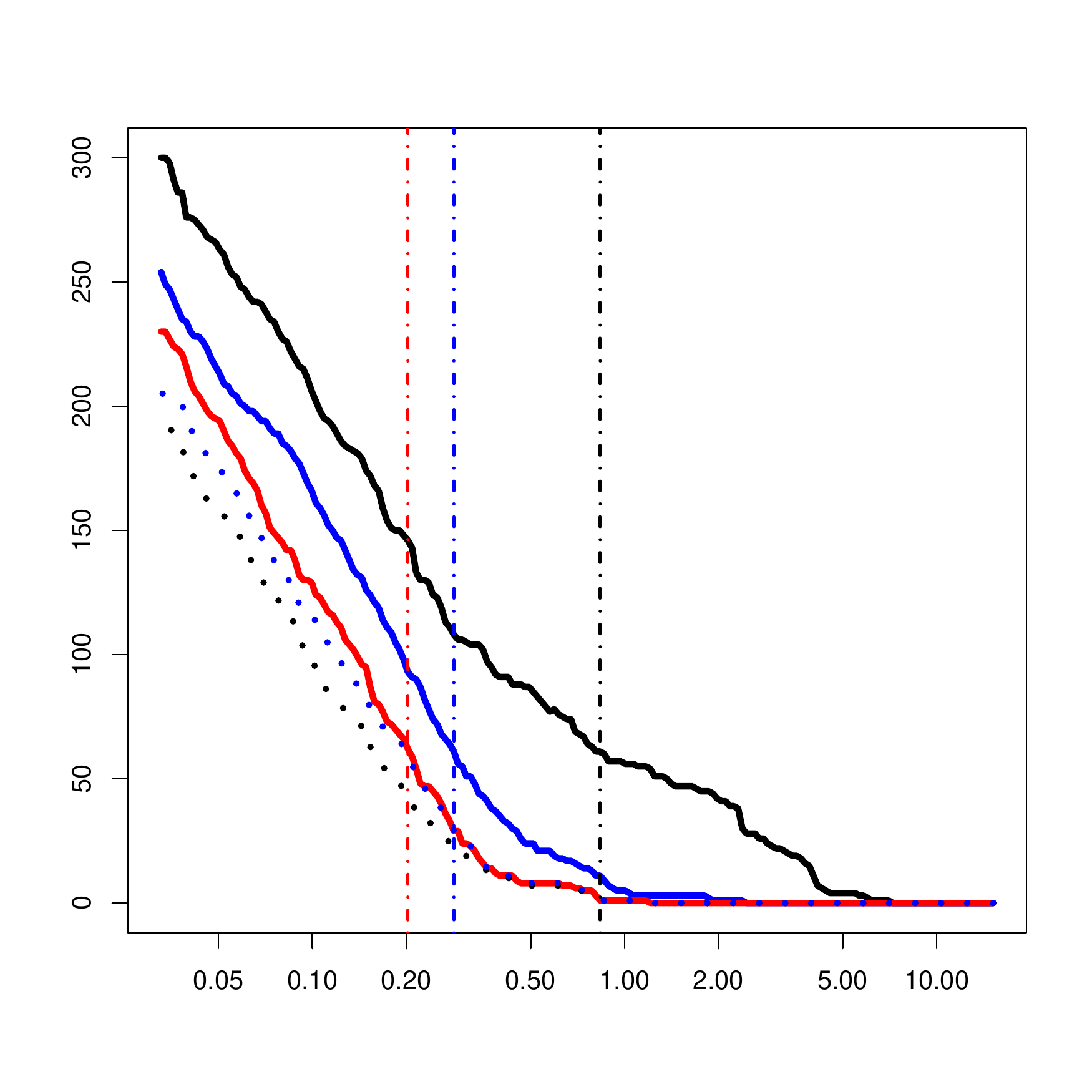} \\
    \rotatebox{90}{\hspace{1.8cm} $\lambda=.28$} &
    \includegraphics[width=.32\textwidth]{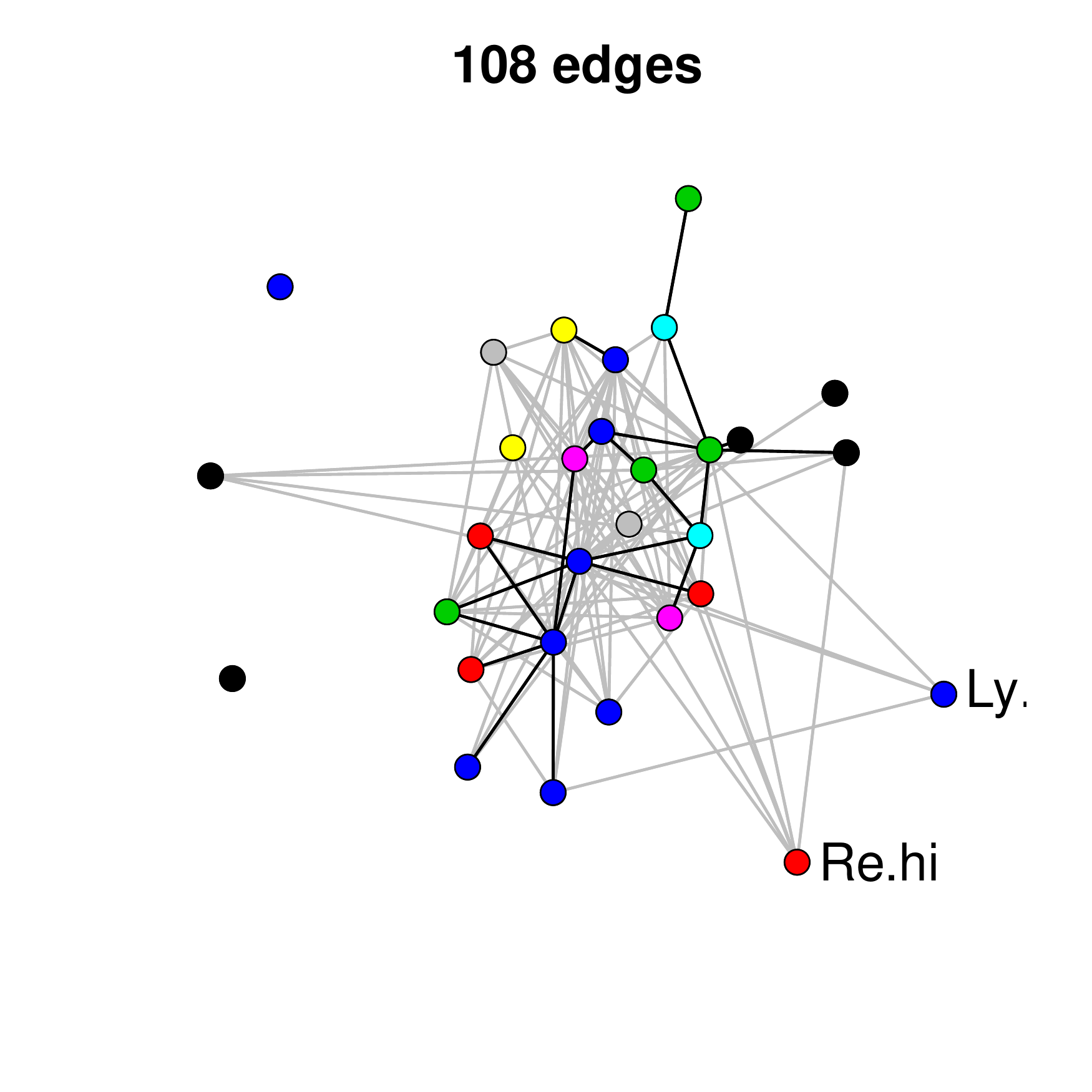} & \includegraphics[width=.32\textwidth]{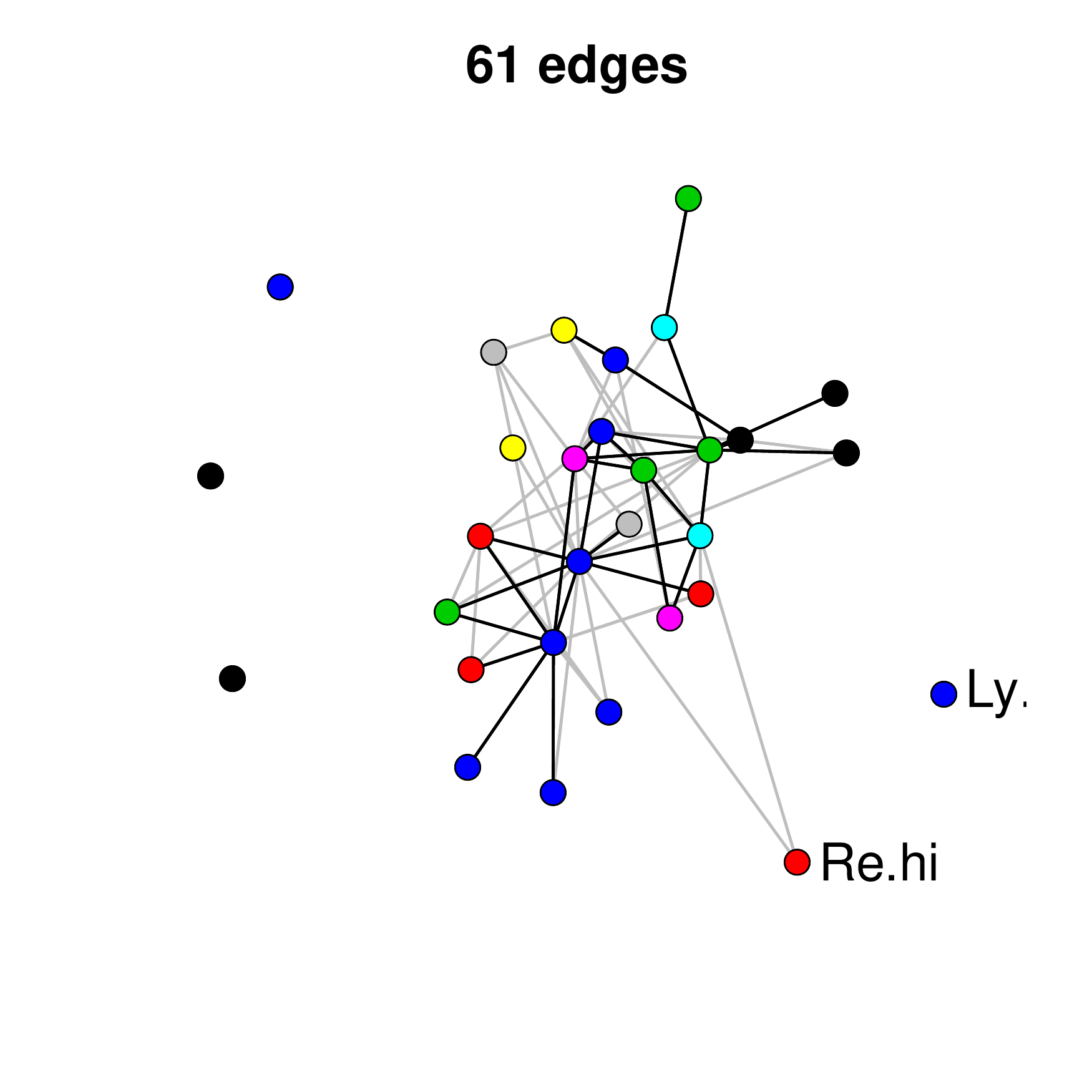} &
    \includegraphics[width=.32\textwidth]{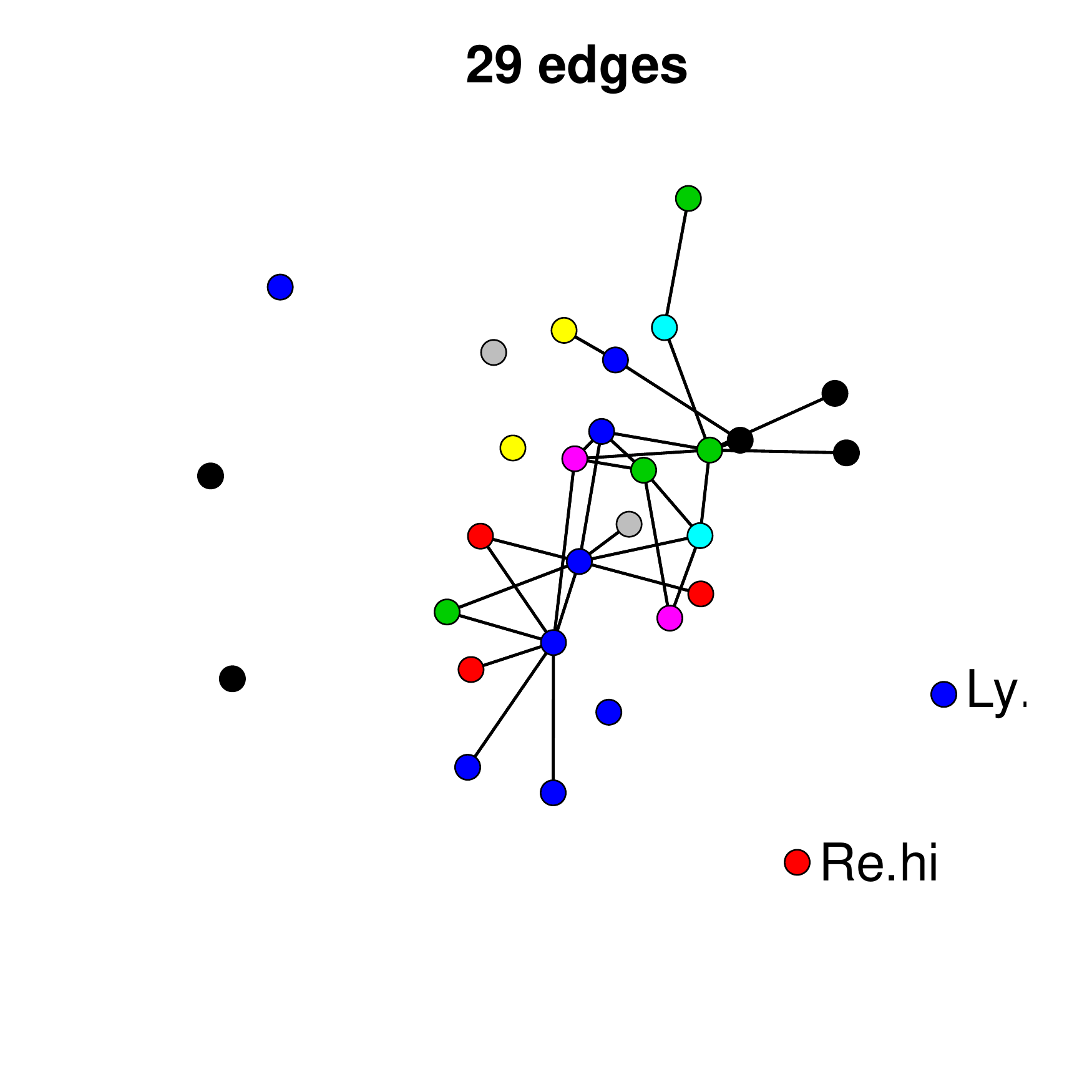} \\
    \rotatebox{90}{\hspace{1.8cm} $\lambda=.20$} &
    \includegraphics[width=.32\textwidth]{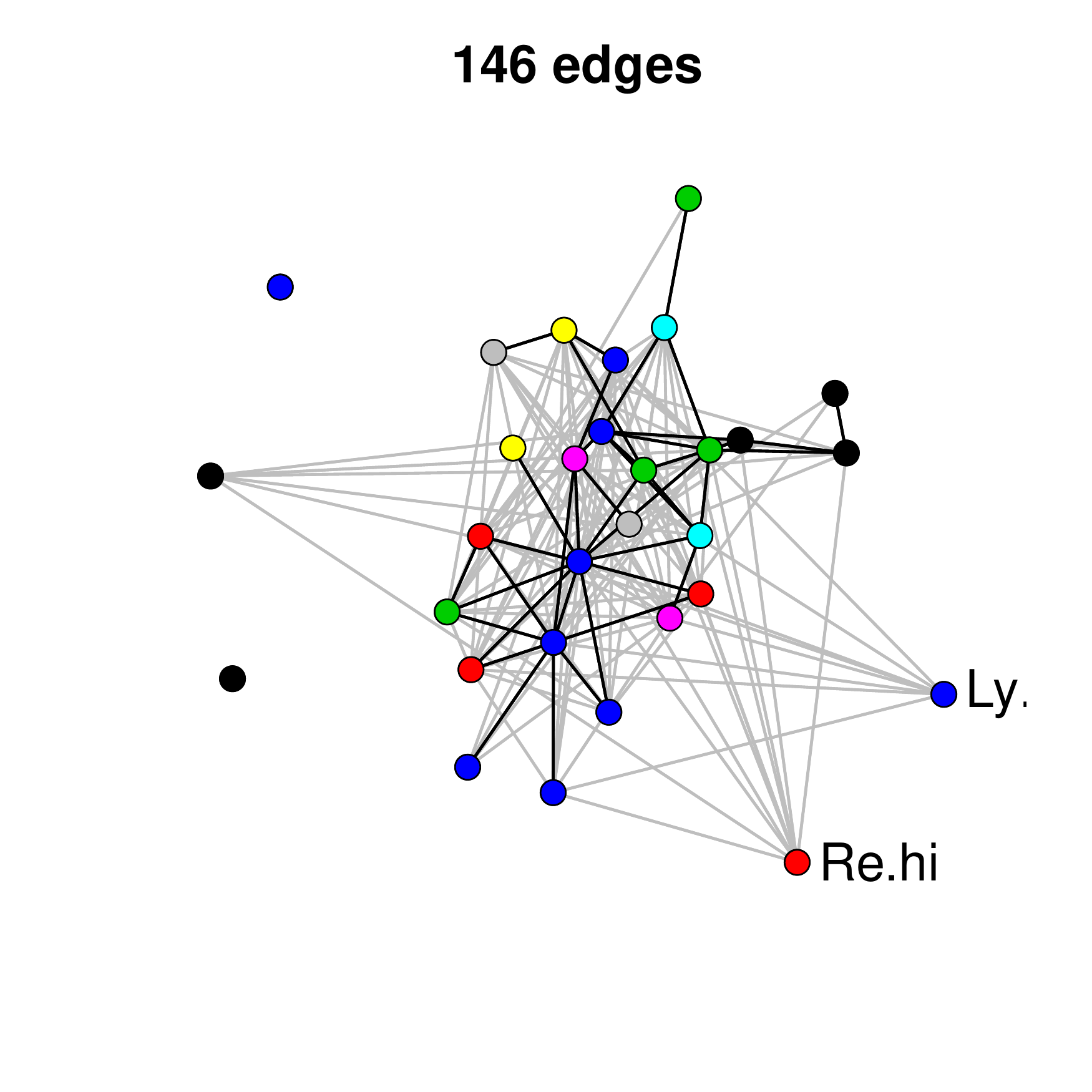} &
    \includegraphics[width=.32\textwidth]{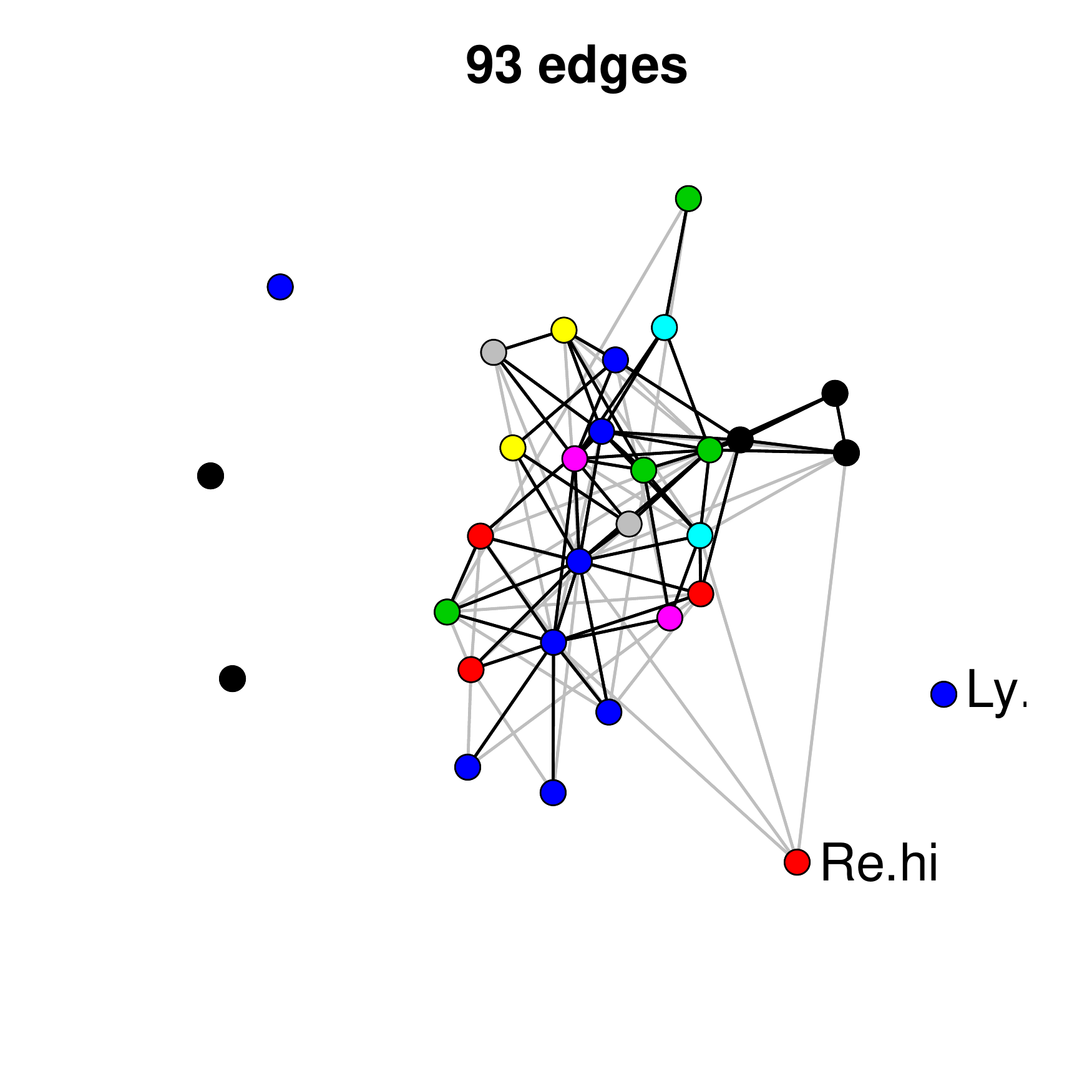} &
    \includegraphics[width=.32\textwidth]{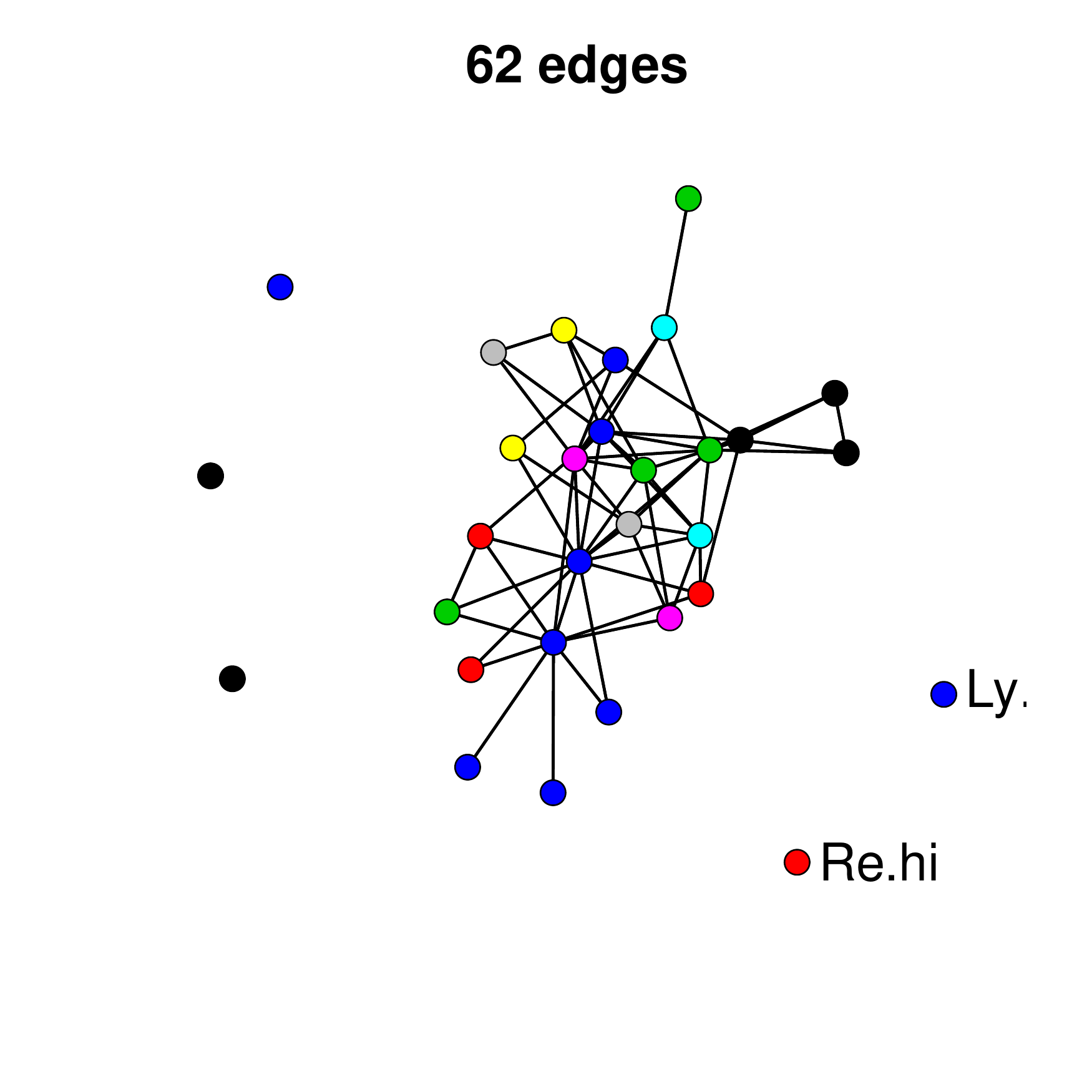}
  \end{tabular}
  \caption{
  Inferred networks with increasing penalties (top: $\lambda=.84$, middle: 
$\lambda=.28$, bottom: $\lambda=.20$) for different covariate sets (see top 
line). Each node corresponds to a given species. The position of the nodes are 
kept fixed. Node color: species family \citep[see][]{FNA06}. Black edges: edges in 
common with the network inferred with all covariates and same $\lambda$. The 
missing network (top right panel, all covariates and $\lambda=.84$) contains 
only one edge. 
Top right: number of edges as a function of $\lambda$. Black: no
covariate, blue: temperature and depth, red: all covariates, dotted
black: common edges with no and all covariates, dotted blue: common
edges between two and all covariates.  Vertical dashed lines: the
three chosen values of $\lambda$. \label{fig:barents-networks}}
 \end{center}
\end{figure}



\subsection{French Presidential Elections, 2017} \label{sec:election}

Our second dataset comes from the first round of the French
presidential election of 2017 and consists in the votes cast for each
of the 11 candidates in the more than 63 000 polling stations. Our
goal here is to find \emph{competing} candidates, who appeal to
different voters, and \emph{compatible} candidates, who appeal to the
same voters, after accounting for the fact that elections are a
zero-sum game.

Data were downloaded from the French open data platform
\url{data.gouv.fr}\footnote{\url{ https://www.data.gouv.fr/fr/datasets/election-presidentielle-des-23-avril-et-7-m
    ai-2017-resultats-definitifs-du-1er-tour-par-bureaux-de-vote/}}
and filtered to remove stations with no votes. To reduce inference
times, we consider a random subset of 13,704 stations that accounted
for 20\% of the registered population. The voting population in those
booths varied wildly, ranging from 10 to 105,891 (6th district of
French citizens living abroad) registered voters, with a median at 736
and 99.5\% of the stations with less than 1,700 voters. We consider
the log-registered population of voters, and not log-turnout, as
an offset to account for different station sizes. This means in
particular that votes are not affected by the \emph{compositionality
  effect} as much as in other settings, as they do not sum up to the
offset.  Voting patterns are well-known to depend on geography and we
therefore consider department (a French administrative division) as a
proxy for geography.

We consider three models in total: without offset, with offset but no
covariate, with offset and covariates and use the same grid of
$\lambda$ -- decreasing geometrically from 1 to $1e^{-3}$ in 31 steps --
for all. The optimal value $\lambda^\star$ was selected using StARS
with 100 subsamples of size 1170 ($\simeq 10\sqrt{n}$). Results of our
analysis are displayed in Figure~\ref{fig:french-elections}.
\begin{figure}[htbp!]
  \centering
  \begin{tabular}{lc@{\hspace{0.25cm}}c@{\hspace{0.25cm}}c}
    & No offset & Offset & Offset + Departments \\
    \rotatebox{90}{\hspace{1.75cm} \small Networks} & 
\includegraphics[width=.3\linewidth]{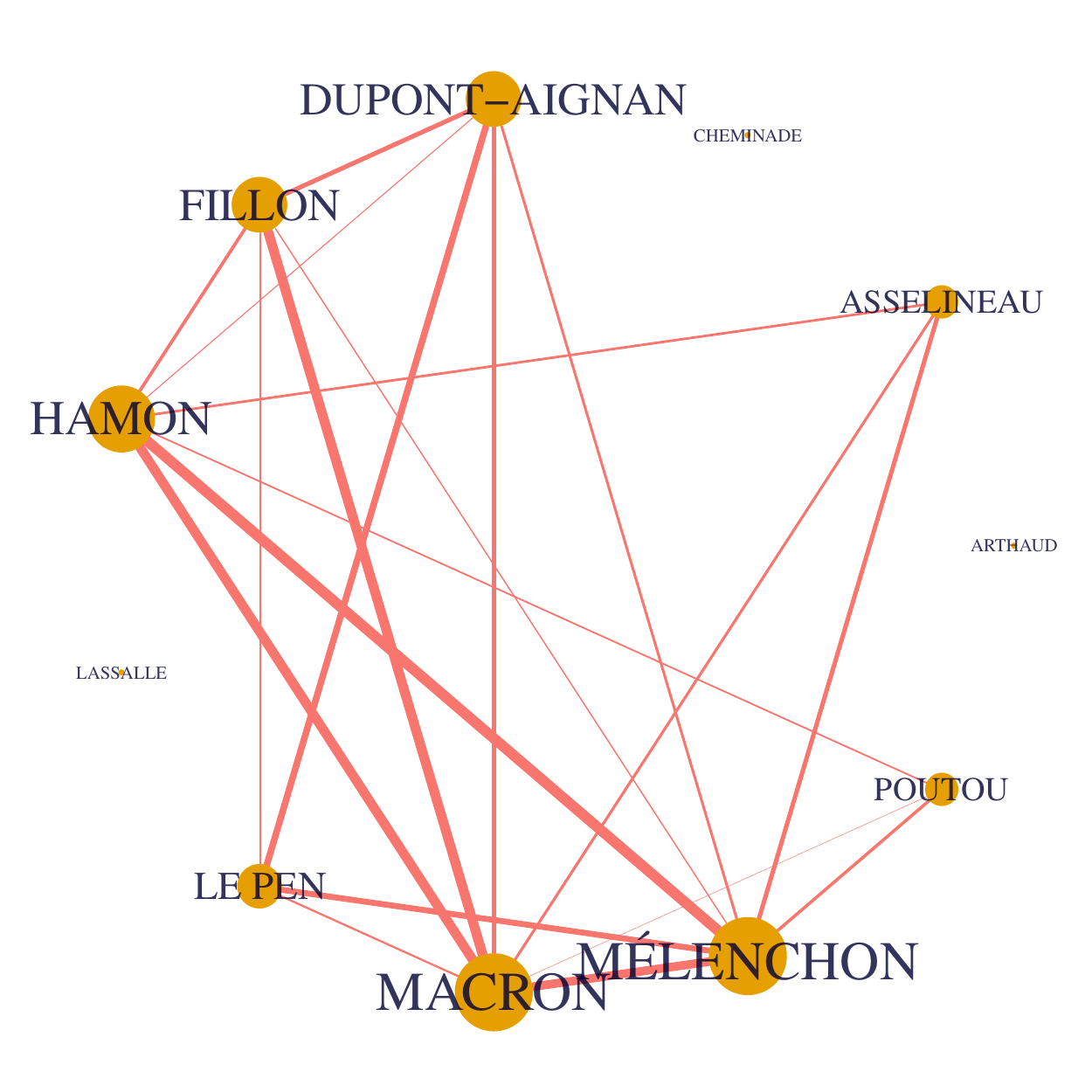} &  
\includegraphics[width=.3\linewidth]{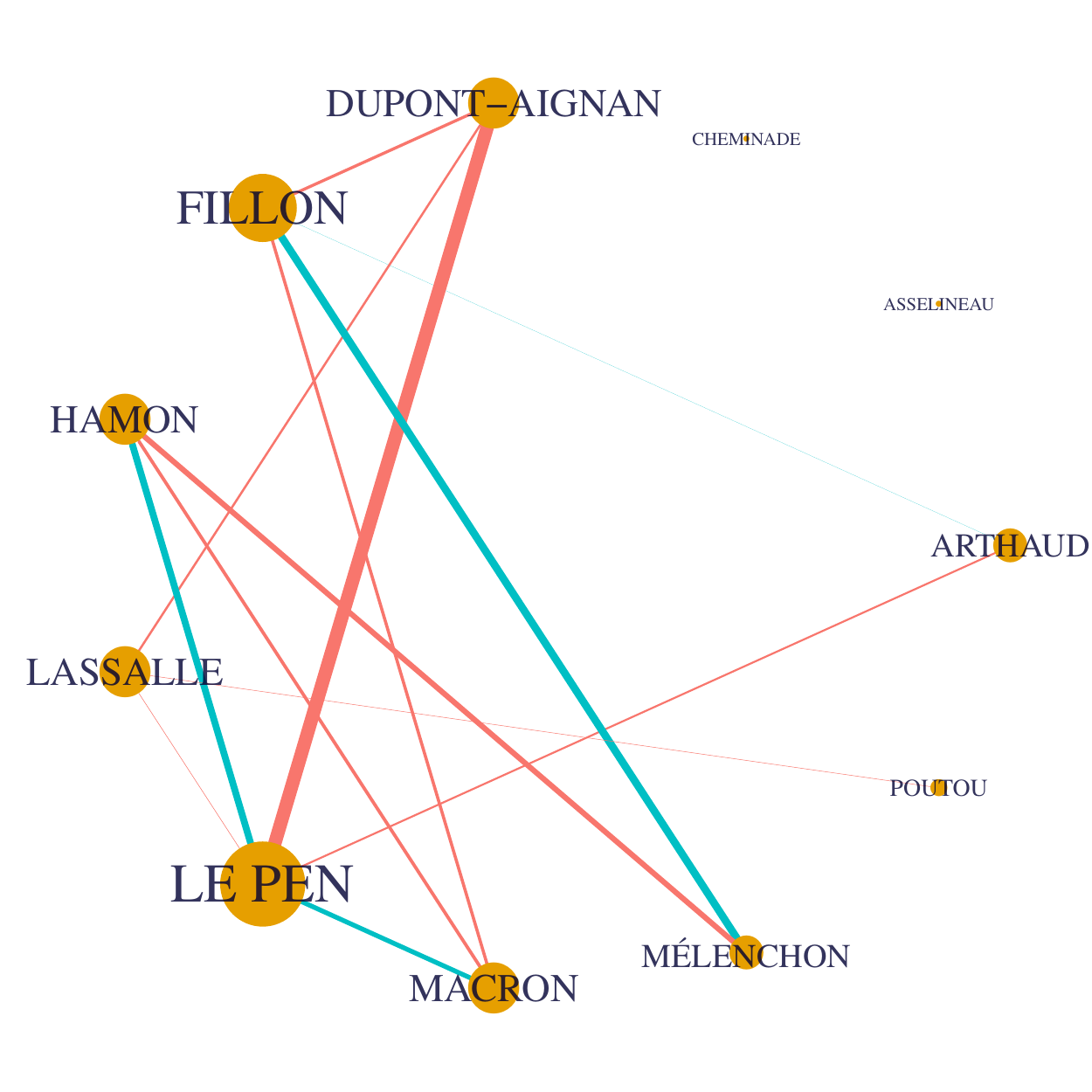} &  
\includegraphics[width=.3\linewidth]{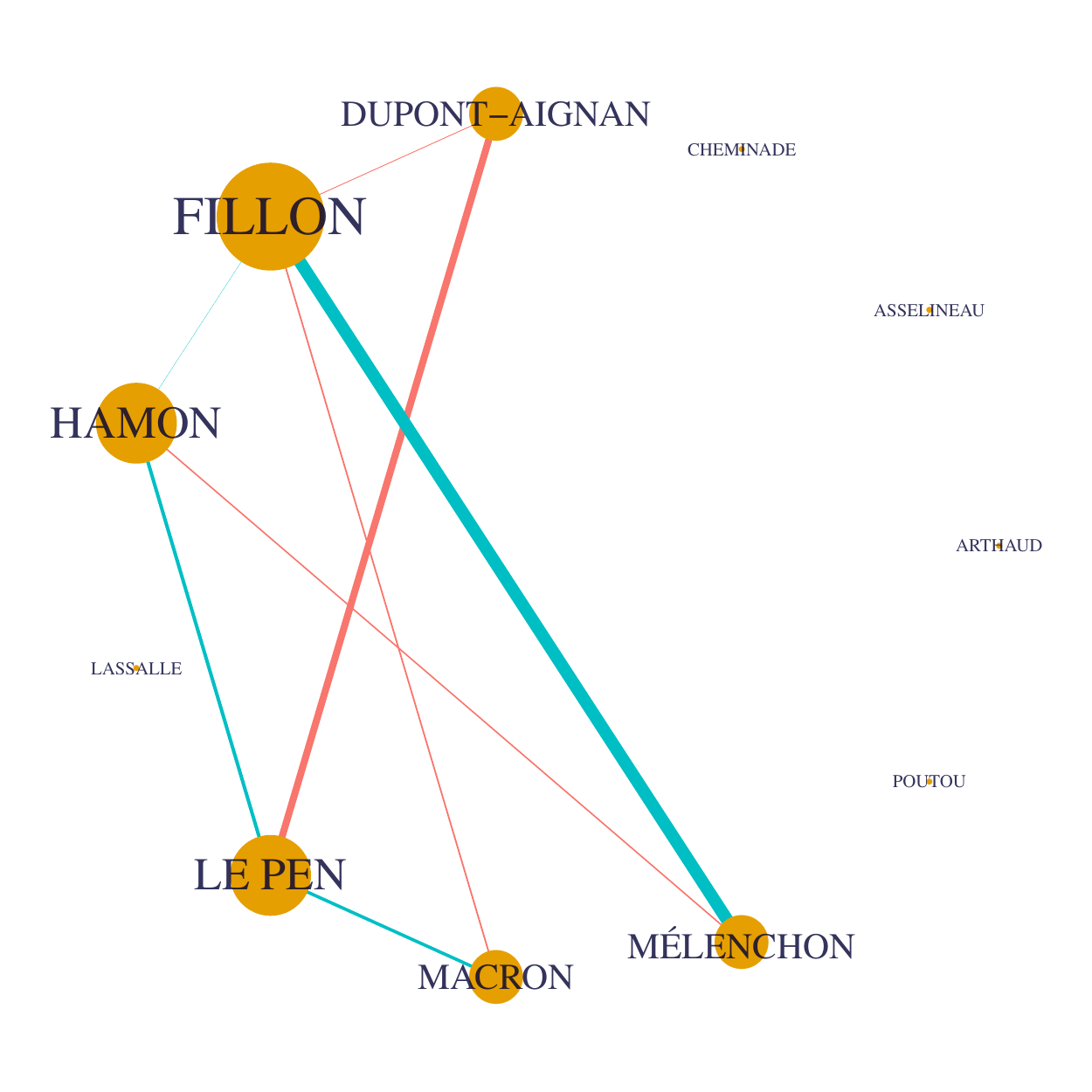}  
\\
    \small \rotatebox{90}{\hspace{.15cm} Latent Positions $\mb{M}$ (PCA)} &  
\includegraphics[width=.3\linewidth, 
keepaspectratio]{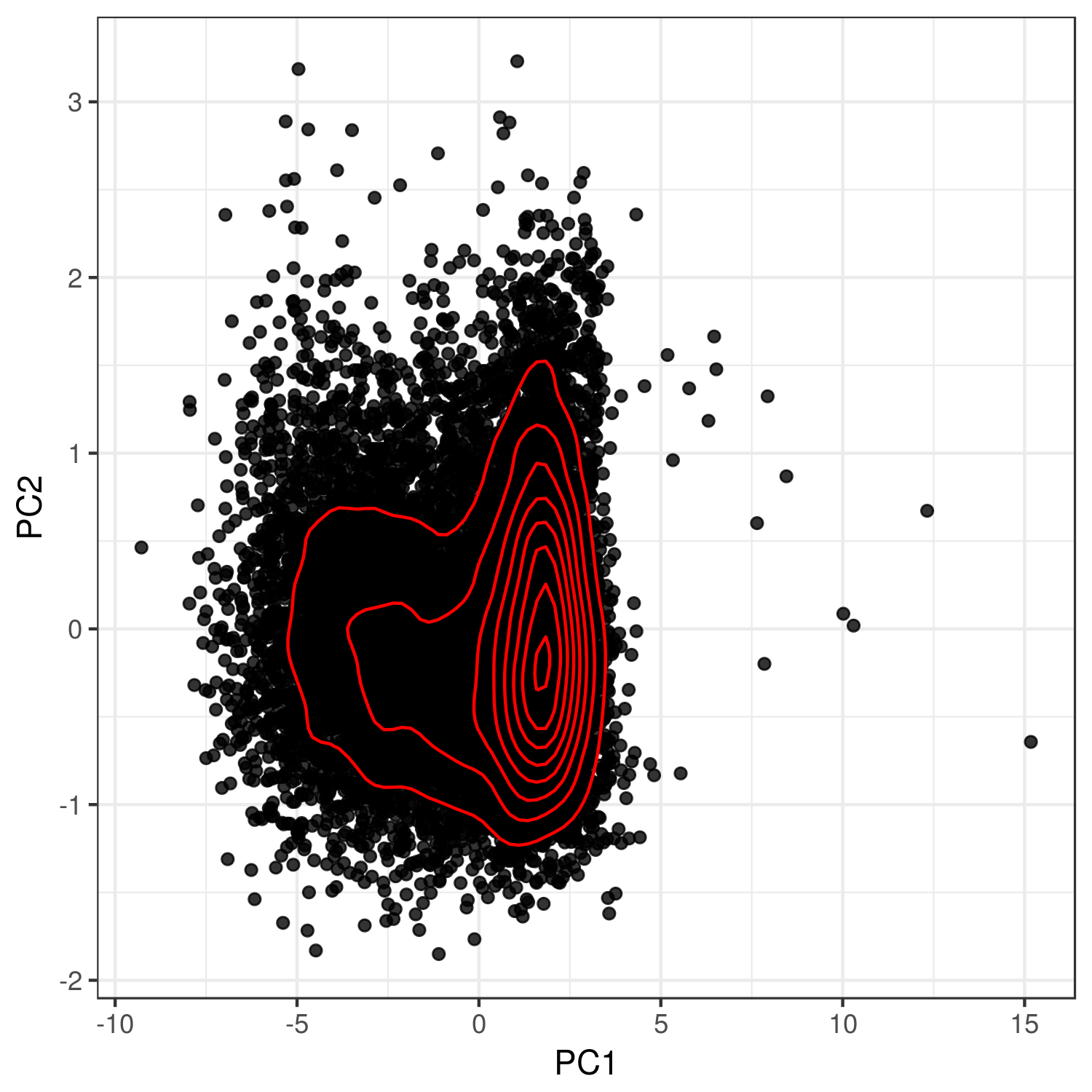} &  
\includegraphics[width=.3\linewidth, 
keepaspectratio]{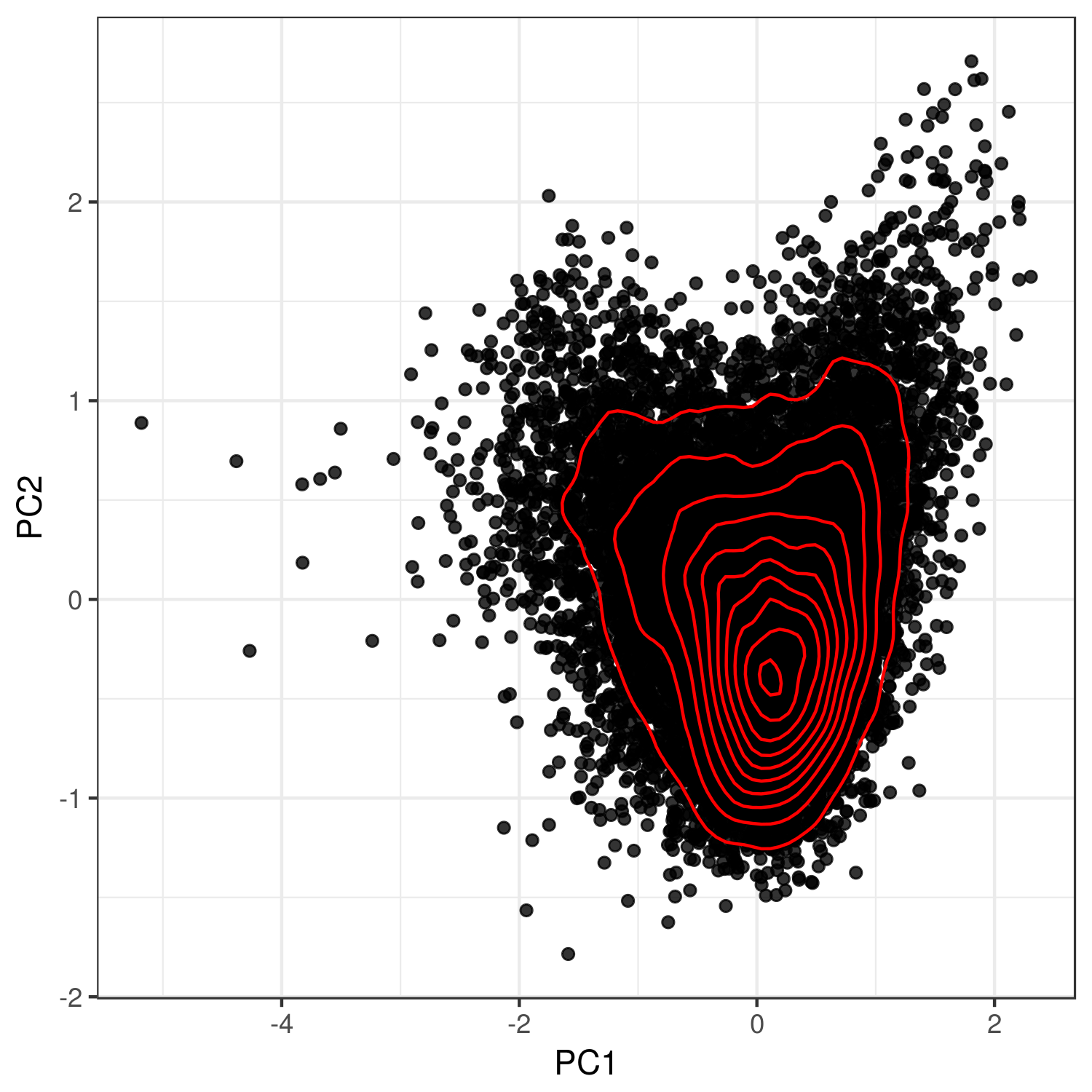} &  
\includegraphics[width=.3\linewidth, 
keepaspectratio]{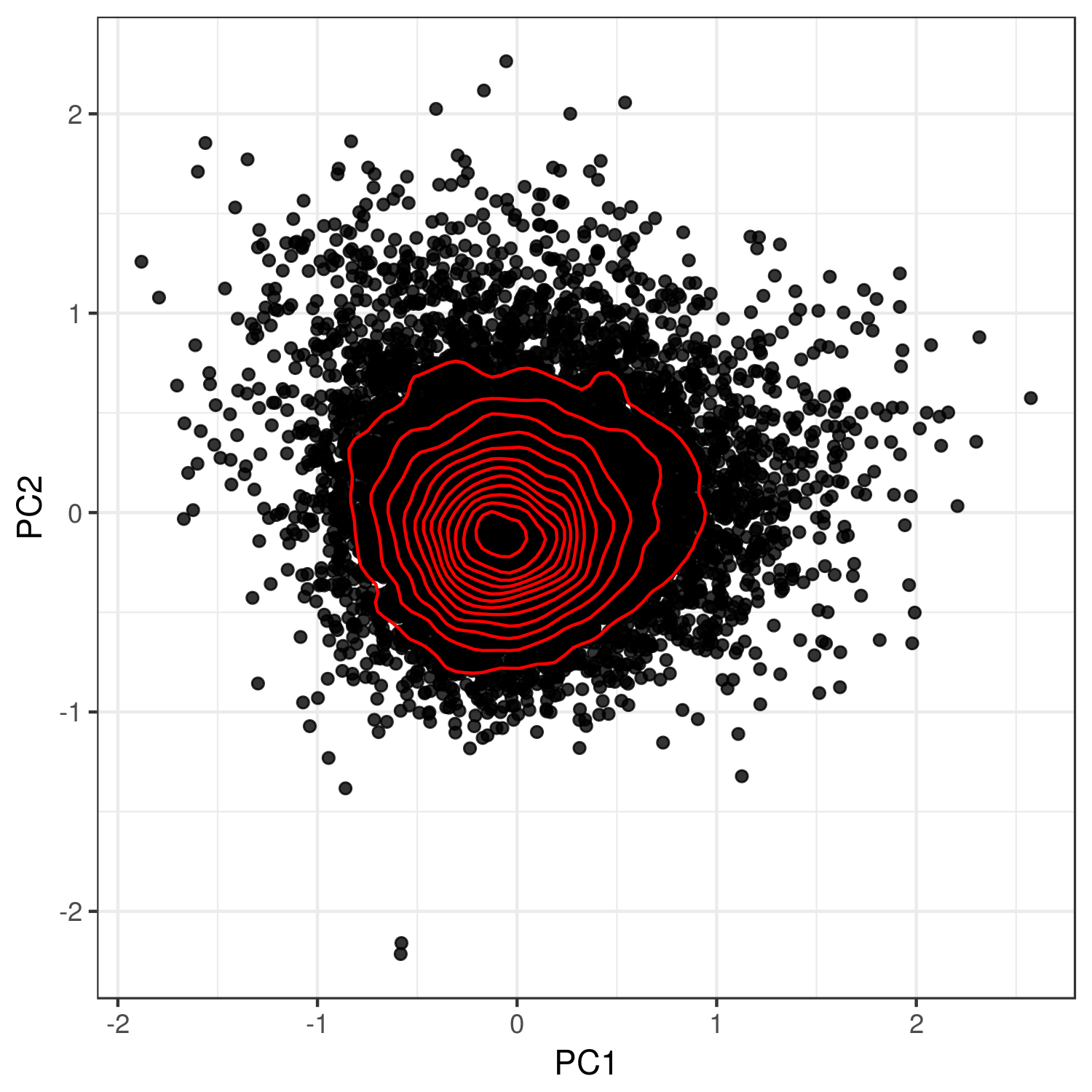}  \\
  \end{tabular}
  \caption{\label{fig:french-elections} Network between candidates of
    last French presidential elections. Top row: Networks inferred
    under different models. Edges represent partial correlations
    $\rho_{ij}$: their thickness is proportional to $|\rho_{ij}|$ and they are
    colored red if $\rho_{ij} > 0$ and blue if $\rho_{ij} < 0$. A
    node's size and label size are proportional to its degree. Bottom
    row: Positions $\mb{M}$ of the polling stations in the Gaussian
    latent space. Since the latent space has dimension $11$, we
    performed a PCA of $\mb{M}$ and show only the principal plane. Red
    lines represent contour lines of the density estimated with a 2D
    kernel. The density should be bivariate Gaussian and departures
    from elliptic curves reveal remaining structure not accounted for
    in the model.}
\end{figure}

\paragraph{The offset matters.} Figure~\ref{fig:french-elections} shows that 
the inclusion of an offset drastically reduces the density of the reconstructed 
network and alters the sign and strength of partial correlations. Failing to 
accoung for varying station sizes leads to a spurious positive partial 
correlations between most candidates: the shift of all stations towards the 
positive orthant in the latent space are mistaken for positive correlations 
between all coordinates. The offset counteracts this by translating back all 
stations towards the origin along the direction $\mathbb{R} \mb{1}$. 

\paragraph{Correcting for geography is 
important.} Figure~\ref{fig:french-elections} also shows that correcting for 
geography also changes the graph but to a lesser extent. However, when we move 
back to the Gaussian latent space and examine the latent positions of the 
polling stations ($\mb{M}$), we do not observe the expected elliptic 
distribution of a multivariate Gaussian (Fig.~\ref{fig:french-elections}, 
left panel). Taking the department of origin into account helps recover 
ellipticity and confirms that geography is indeed a strong structuring factor in 
the latent space. 

\paragraph{Political interactions.} If we consider the network reconstructed 
with the offset and geographic covariate as the most reliable, results show 
that candidate with similar political leaning appeal to the same voters (M. Le 
Pen and N. Dupont-Aignan (both far right), B. Hamon (left) and J.-L. Mélenchon 
(far left), E. Macron (center) and J.-F. Fillon (right)) whereas candidate 
with different leanings appeal to different voters (M. Le Pen versus B. Hamon 
and E. Macron, J.-F. Fillon versus J.-L. Mélenchon). More precisely, a negative 
partial correlation between candidates A and B means, all other things being 
equal, that a high vote for one candidate in a station is correlated to a low 
vote for the other. 

This may explain the absence of negative correlation between far left and far 
right: although their electorates may differ, they vote in the same stations. 
Similarly, the fact that the positive partial correlation between E. Macron and 
B. Hamon disappears when controlling for geography means that they have high 
voter shares in the same departments but not necessarily in the same polling 
stations. This is confirmed by the high correlation ($0.76$) of their 
respective regression coefficients across departments. 


\subsection{Oak mildew} \label{sec:mildew}

The metagenomic dataset introduced in \cite{JFS16} consists of
microbial communities sampled on the surface of oak leaves (the
samples). The leaves were collected on trees with different resistance
levels to the fungal pathogenic species \textsl{E. alphitoides},
responsible for the oak powdery mildew. Table \ref{tab:OTU} provides
the available classification information about the bacterial and
fungal OTU appearing in at least one network inferred in our
analysis. Unfortunately, not all OTU can been identified at the
species level and some OTU are not related to any known species.  In
the following, we consider two groups of samples labeled by
\citeauthor{JFS16}: $n_r = 39$ resistant samples (where
\textsl{E. alphitoides} was essentially absent) and $n_s = 39$
susceptible samples (where a significant activity of
\textsl{E. alphitoides} was detected).  In addition to the sampling
tree, several covariates, all thought to potentially structure the
community, were measured for each leaf: orientation, distance to
trunk, distance to ground, distance to base. After sequencing,
clustering into operating taxononomic units (or OTU -- a proxy for
species), and a final filtering of the identified bacterial and fungal
communities with too few reads, the total number of species considered
is $p = 114$ OTUs in this data set (66 bacterial ones and 48 fungal
ones, including \textsl{E. alphitoides}).

\begin{table}[ht]
\centering
\begingroup\small
\begin{tabular}{lllll}
  \hline
Type & OTU &  Family & Genus & Species \\
\specialrule{.1em}{.05em}{.5em}
Fungi &  f1 & Dermateaceae & Naevala & Naevala minutissima \\ 
&  f3 & -- & -- & -- \\ 
&  f4 & Erysiphaceae & Erysiphe & Erysiphe hypophylla \\ 
&  f8 & Hyaloscyphaceae & Catenulifera & Catenulifera brevicollaris \\ 
&  f10 & -- & -- & -- \\ 
&  f12 & Amphisphaeriaceae & Monochaetia & Monochaetia kansensis \\ 
&  f17 & Herpotrichiellaceae & Cyphellophora & Cyphellophora hylomeconis \\ 
&  f19 & -- & -- & -- \\ 
&  f25 & unidentified & Cryptococcus & Cryptococcus magnus \\ 
&  f27 & unidentified & Strelitziana & Strelitziana mali \\ 
&  f29 & Mycosphaerellaceae & Xenosonderhenia & Xenosonderhenia syzygii \\ 
&  f32 & -- & -- & -- \\ 
&  f39 & -- & -- & -- \\ 
&  f1085 & Mycosphaerellaceae & Mycosphaerella & Mycosphaerella marksii \\ 
&  f1090 & Herpotrichiellaceae & Cyphellophora & Cyphellophora hylomeconis \\ 
&  f1278 & Mycosphaerellaceae & Mycosphaerella & Mycosphaerella punctiformis \\ 
&  Ea & Erysiphaceae & Erysiphe & Erysiphe alphitoides \\ 
Bacteria & b13 &  Oxalobacteraceae & -- & -- \\ 
&  b153 &  Oxalobacteraceae & -- & -- \\ 
&  b21 &  Pseudomonadaceae &  Pseudomonas & -- \\ 
&  b25 &  Enterobacteriaceae & -- & -- \\ 
&  b26 &  Oxalobacteraceae & -- & -- \\ 
&  b33 &  Microbacteriaceae &  Rathayibacter & -- \\ 
&  b364 &  Oxalobacteraceae & -- & -- \\ 
&  b37 &  Beijerinckiaceae &  Beijerinckia & -- \\ 
&  b44 & -- & -- & -- \\ 
&  b60 & -- & -- & -- \\ 
   \hline
\end{tabular}
\endgroup
\caption{Type of microorganism (bacteria or fungi) and higher level taxonomic 
assignments (family, 
genus and species) of the 27 operational taxonomic units (OTUs) interacting in 
the inferred microbial networks. Unknown assignments at a given rank are 
reported as '--'. \label{tab:OTU}} 
\end{table}

Our aim here is to unravel the association between the different
microbial and fungal species by reconstructing the ecological
network. Obviously, we are especially interested in the interactions
between \textsl{E. alphitoides} and the other species. We emphasize
that unlike \texttt{SPiEC-Easi} or \texttt{sparCC}, that are limited
to interactions between bacteria or between fungi due their
normalisation step, we can actually investigate interactions between
bacterial and fungi \textsl{E. alphitoides} although the
sequencing depths differ for each type. A similar target was already at the
core of \citeauthor{JFS16}'s work. However our approach differs from a
methodological view-point as we jointly estimate the effect of the
covariates $\mb{B}$ and the dependency structure $\mb{\Omega}$ while
they only corrected the observed counts for the effect of the
covariates using a regression model before feeding the residuals from
that regression to a network inference method. This two-steps
procedure fails to account for the fact that $\mb{B}$ is estimated and
to propagate uncertainty from the first step to the second
one. Moreover, \citeauthor{JFS16} focused their study on the set of
susceptible samples, while we propose here to infer three networks:
one for susceptible samples, one for resistant samples and one for
merged samples. By these means, we hope to obtain a more thorough map
of interactions between the pathogen and its ecosystem.

The three PLN models respectively including the susceptible, the
resistant and both samples were defined as follows: for the
susceptible and the resistant models, we applied \texttt{PLNnetwork}
by including simple effects of the orientation and of the distance to
the trunk (the distances to the ground and to the base were highly
correlated with the former, and we used it as a representative of
these three covariates). For the model merging all the samples, we
added the covariate describing the tree status (either resistant or
susceptible), with both simple effects and interactions with the two
other covariates (orientation and distance to trunk). These two
approaches -- separating or merging the samples -- address different
yet complementary goals: by separating the samples, we assume that the
two underlying networks (and thus covariances) are different and need
a specific analysis; the counterpart that merges all samples aims to
render a synergistic network that encompasses important interactions
from both situations after correction of the mean effects due to the
tree status (resistant or susceptible).

\def\netsize{.5\textwidth}
\begin{figure}[htbp!]
  \centering

  \begin{tabular}{@{}cc@{}}
    {\small resistant samples} & {\small susceptible samples} \\ 
    \includegraphics[width=\netsize]{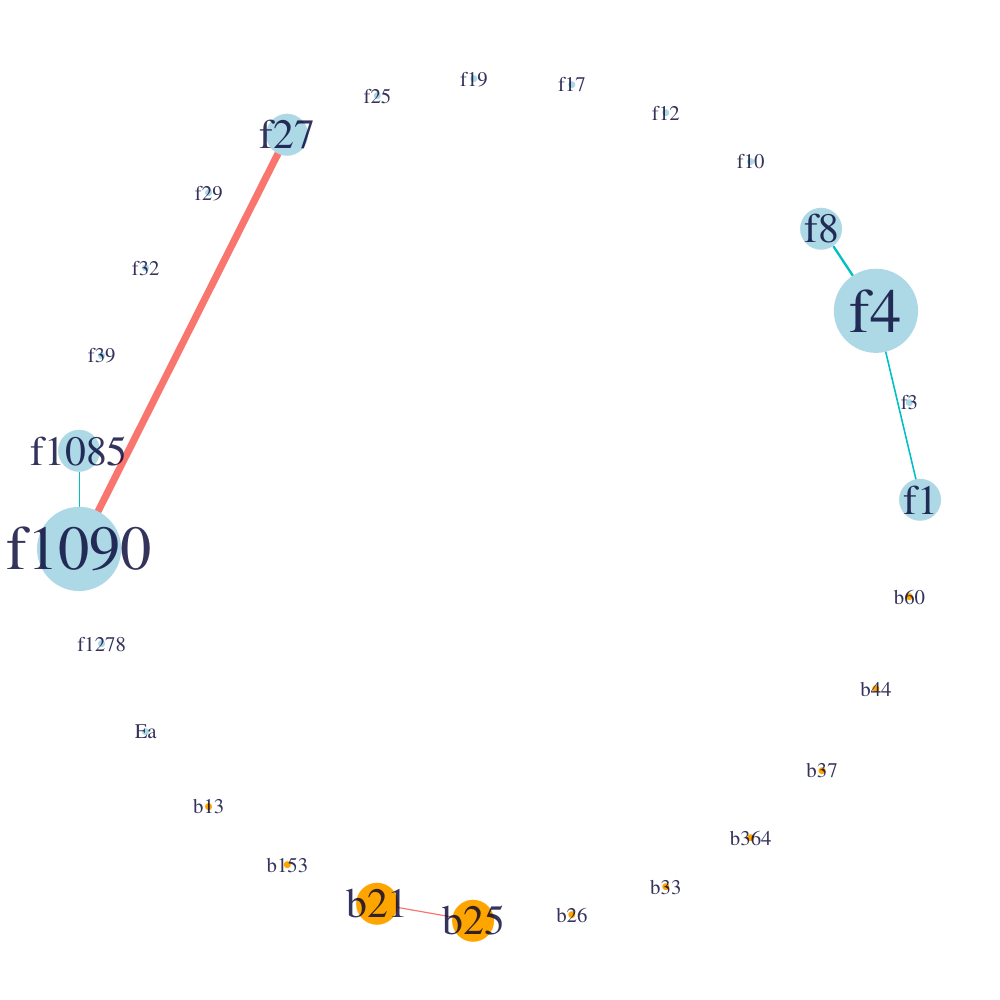}
  & \includegraphics[width=\netsize]{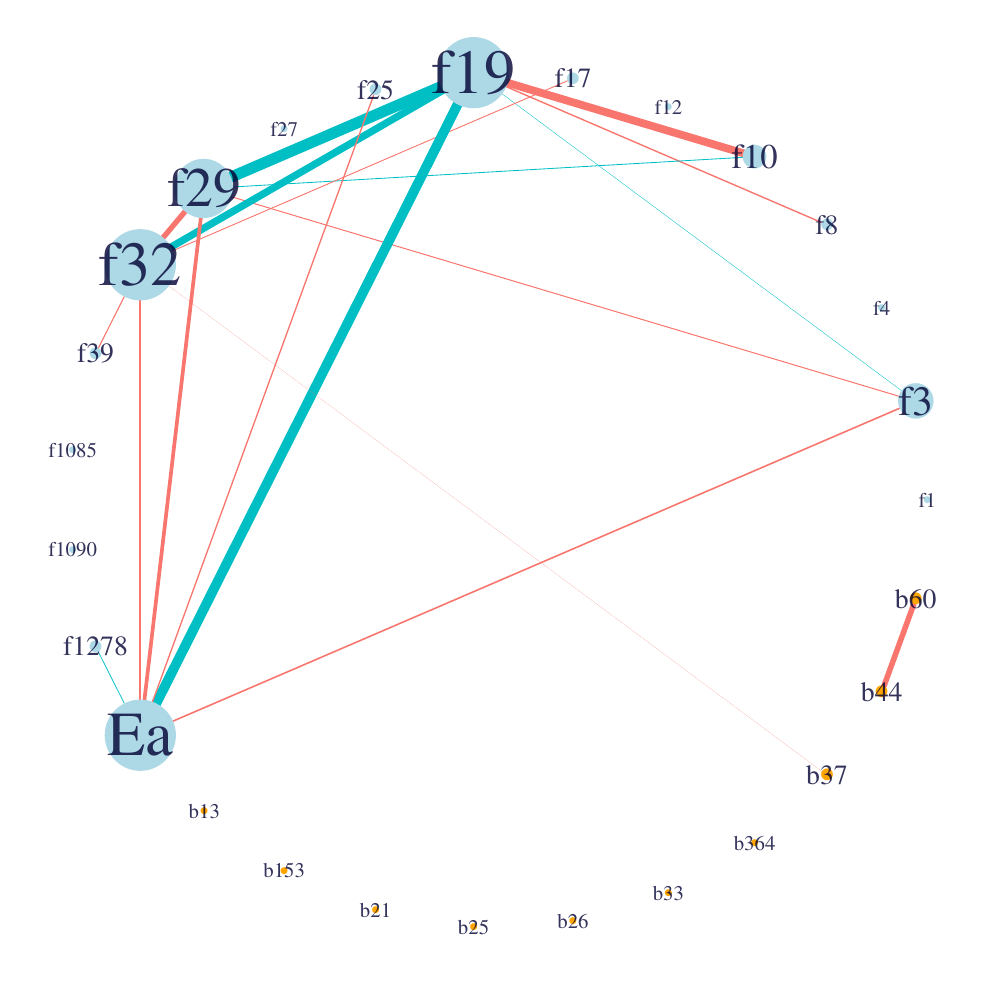} \\
   {\small samples from both origins} & {\small regression coefficient for orientation} \\
    \includegraphics[width=\netsize]{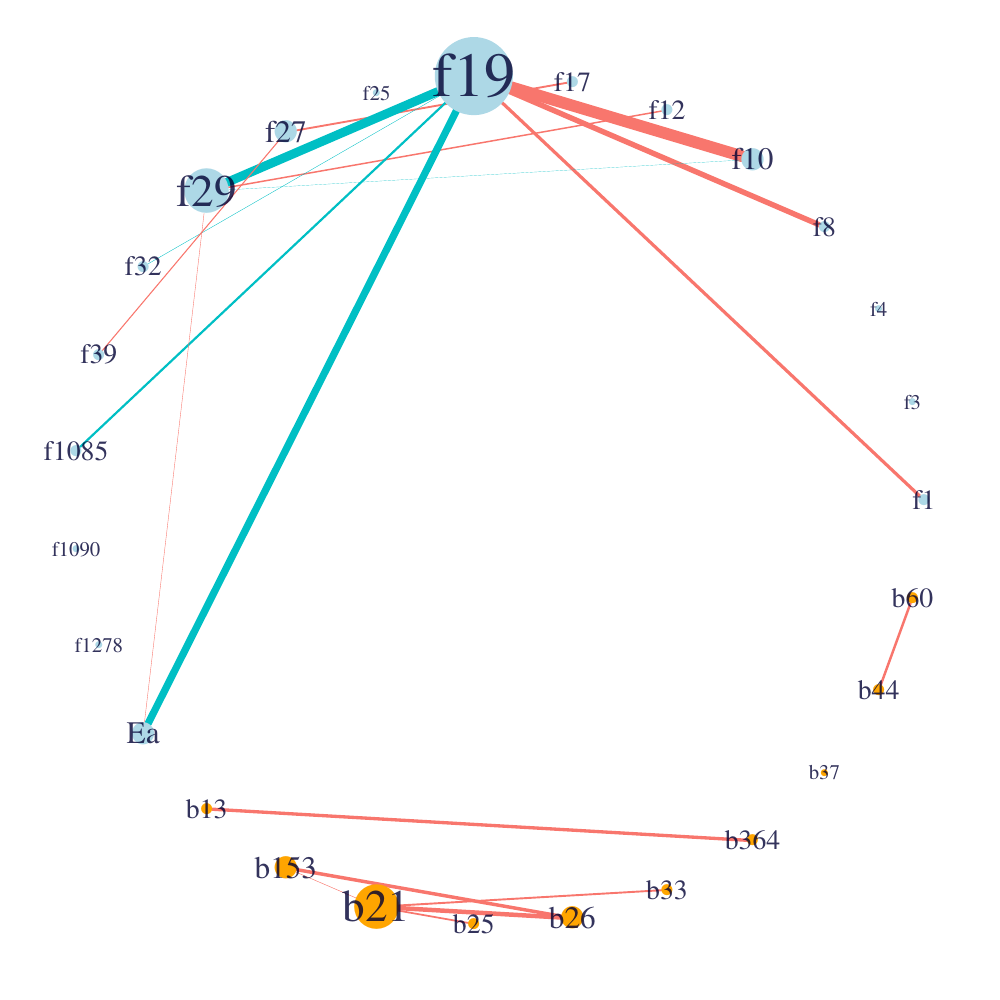}
  &   \includegraphics[width=.5\textwidth]{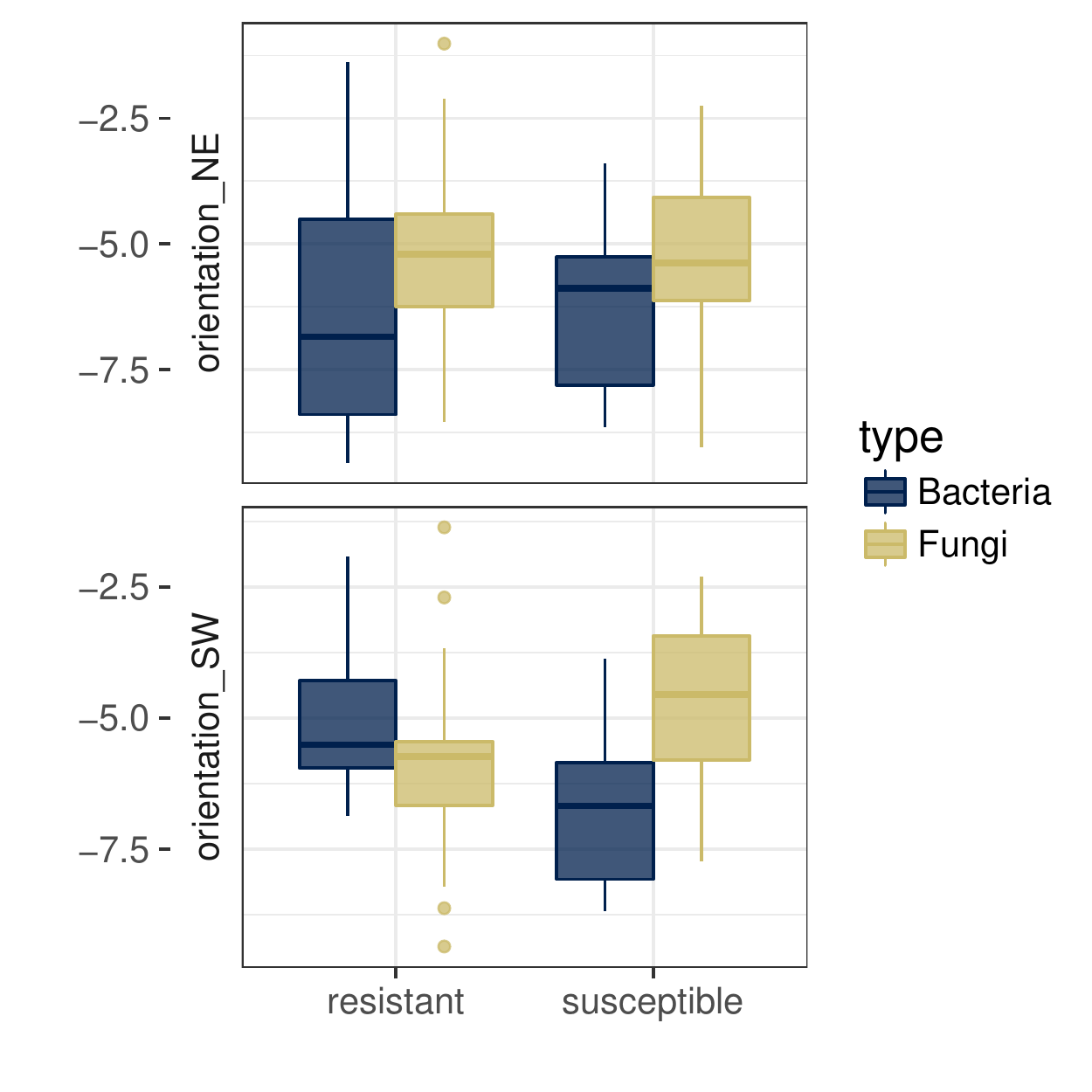} \\     
  \end{tabular}

  \caption{Oak mildew network analysis: networks inferred by
    \texttt{PLNnetwork} and selected by StARS for a stability of
    $0.995$.  Each network correspond to networks inferred using
    samples respectively from the resistant tree, the susceptible tree
    and when merging samples from both origins. Blue vertices
    represent Fungi; orange vertices represent bacteria. Edges
    represent partial correlations $\rho_{ij}$: edge thickness is
    proportional to $|\rho_{ij}|$ and are colored red if
    $\rho_{ij} > 0$ and blue if $\rho_{ij} < 0$. A node's size and
    label size are proportional to its degree. Only nodes having at
    least one edge among the three networks are included in the
    plots. Inset: Boxplot of regression coefficient of abundances against
    orientation.}
  \label{fig:network_mildew}
\end{figure}

Before getting into the interpretation of the results in terms of
species interactions, we remind that the PLN models also enables to
measure the effect of the covariates on each species. The bottom right
panel of Figure \ref{fig:network_mildew} displays the distribution of
the regression parameters of the two orientation indicators (NE $=$
north-east and SW $=$ south-west), in each tree, across each species
type. We do not discuss extensively these results but one may observe
a strong interaction between SW orientation and tree type on both
fungi and bacteria: bacteria are notably depleted in leaves facing SW
in susceptible trees.

We now focus on the results of our analysis in terms of networks in
Figure~\ref{fig:network_mildew}.  All networks inferred with
\texttt{PLNnetwork} where selected with StARS on a 50-size grid of
penalties, using a high stability level of $1-2\beta = 0.995$ to
drastically limit the number of false positive edges. The top row
displays the resistant and susceptible networks, showing very
different patterns, while the consensus network seems to catch
features from both of them.  In the susceptible network, \textsl{E.
  alphitoides} is identified as ($i$) antagonist to fungi f1278, from
the \textit{Mycosphaerella punctiformis} species, which colonizes
living oak leaves asymptomatically and may prevent infection by
\textsl{E. alphitoides} and ($ii$) mutualist to fungi f29, from the
\textit{Xenosonderhenia syzygii} species, usually found in leaf spots,
common on weakened and senescent leaves. The other mutualists of
\textsl{E. alphitoides} unfortunately belong to unknown species and no
similar observations can be made.  Interestingly, in the susceptible
network, the pathogen has less interactions than fungi f19, but is
connected to it, whereas both have few connections in the resistant
network.  As \textsl{E. alphitoides} is known to be responsible for
the mildew disease, the comparison of these networks suggest that its
pathogenic effect is partially mediated by f19. In addition to the
direct effect of the pathogen on a small set of species, its
(negative) effect on fungi f19, which seems to play a central role in
the phyllosphere, leverages its impact on the whole system.  Finally,
the consensus network encompassing both sources of samples resembles
the susceptible network, with some notable discrepancies: a cluster
composed by bacterial species b21, b25, b26, b153 and to a lesser
extent b33 is found in the consensus network, which was only incipient
in the resistant network. This is probably due to the gain in power
induced by a larger sample-size.

%


\paragraph{Acknowledgement} We thank Charlie Pauvert for beta-testing
our implementation and for his feedback. This work was funded by
projects LearnBioControl, Brassica-Dev from the Inra MEM
metaprogramme.

\bibliography{VEM-PLNnetwork}


\end{document}